\documentclass[letterpaper,twocolumn,10pt,anonymous,review]{article}
\usepackage{usenix}

\usepackage{tikz}
\usepackage{amsmath}

\usepackage{filecontents}
\usepackage{mathtools}
\usepackage{filecontents}
\usepackage{subcaption}
\usepackage{comment}
\usepackage{paralist}
\usepackage{url}
\usepackage{xurl}
\usepackage{hyperref}
\usepackage{amsthm}
\usepackage{algorithm}
\usepackage{algorithmic}
\usepackage{booktabs}
\usepackage{xspace}
\usepackage{booktabs} 
\usepackage{multirow}
\usepackage{bm}

\newcommand{\bheading}[1]{\vspace{1pt}\noindent{\textbf{#1}}}

\newcommand{\stash}{\mathsf{stash}}

\newcommand{\prf}{\mathsf{PRF}}
\newcommand{\prp}{\mathsf{PRP}}
\newcommand{\opath}[1]{\mathcal{P}(#1)}
\newcommand{\otro}{\textsf{OTRO}}

\newcounter{note}[section]

\newcounter{notesec}[section]

\newcommand{\ie}{\emph{i.e.}\xspace}
\newcommand{\eg}{\emph{e.g.}\xspace}

\newcommand{\adv}{\mathcal{A}}

\newenvironment{packeditemize}{
\begin{list}{$\bullet$}{
\setlength{\labelwidth}{0pt}
\setlength{\itemsep}{2pt}
\setlength{\leftmargin}{\labelwidth}
\addtolength{\leftmargin}{\labelsep}
\setlength{\parindent}{0pt}
\setlength{\listparindent}{\parindent}
\setlength{\parsep}{1pt}
\setlength{\topsep}{1pt}}}{\end{list}}

\newtheoremstyle{boldthm}
  {3pt}    
  {3pt}    
  {\itshape} 
  {}       
  {\bfseries} 
  {.}      
  { }      
  {}       

\theoremstyle{boldthm}
\newtheorem{theorem}{Theorem}

\newtheorem{lemma}{Lemma}

\theoremstyle{definition}

\theoremstyle{remark}

\begin{document}

\date{}

\title{\Large \bf OTRO: Oblivious Tokenization with Read-Only ORAM}

\author{
{\rm Jonghyun Lee$^{1}$, Yongqin Wang$^{2}$, Rachit Rajat$^{3}$, Daniel Wong$^{4}$, Mengyuan Li$^{1}$, Murali Annavaram$^{1}$}\\
$^{1}$University of Southern California \quad
$^{2}$Roblox \quad
$^{3}$NVIDIA \quad
$^{4}$University of California, Riverside
}

\maketitle

\begin{abstract}
Tokenizers running on CPU create a security risk in large language model (LLM) serving, even when the whole pipeline runs inside CPU and GPU trusted execution environments (TEEs). A tokenizer maps prompt text to token IDs through table lookups, and the access pattern of these lookups leaks information about the prompt. Recent work reconstructs full prompts from tokenizer access patterns on Intel TDX. Closing this leak with a standard tree-based oblivious RAM (ORAM), such as Circuit ORAM, is impractical as it slows the tokenizer by $\sim$19.3$\times$, inflating time-to-first-token (TTFT) by 3.1$\times$ on average, and up to 5.1$\times$.

In this paper, we present \otro{}, an oblivious tokenization and detokenization path that protects access patterns at low cost for production-scale prompts. \otro{} is built on square-root ORAM (SqrtORAM). Each access is cheap, but the rebuild happens in one large, bursty step that costs $O(N \log^2 N)$ time, so a plain, drop-in use of SqrtORAM is impractical. The key insight behind \otro{} is that tokenizer tables are read-only, enabling \otro{} to keep a pool of independently permuted SqrtORAM replicas. While other replicas are serving lookups, one replica can rebuild off the serving path. \otro{} also minimizes rebuilds by increasing each replica's stash beyond the standard $\sqrt{N}$ size. \otro{} further overlaps tokenization with GPU prefill, hiding any remaining rebuild latency behind GPU computation. The overlap lets \otro{} size the replica pool to the rebuild throughput rather than to the model's maximum context length. We implement \otro{} in HuggingFace Tokenizers and vLLM, inside an Intel TDX CVM with an NVIDIA H100 GPU. \otro{} is the first oblivious tokenizer to match plaintext TTFT for production-scale prompts.

\end{abstract}

\section{Introduction}
Large language models (LLMs) are now common in everyday tools and services. They power interactive assistants, chatbots, code agents, and other conversational AI systems.
LLMs often process sensitive prompts, including proprietary documents, personal communications, and security-critical requests. Sending these prompts to a cloud service provider (CSP) without protection raises significant security concerns~\cite{weiss2024your,hayase2024data,gao2025know,zhang2024time,carlini2024stealing}. Keeping user prompts and AI-generated responses \textbf{confidential} during \emph{cloud-based} inference helps maintain user trust and comply with legal and ethical data protection standards.

\bheading{The Emergence of Confidential LLM Serving.}
To address these concerns, LLM serving increasingly uses Trusted Execution Environments (TEEs)~\cite{amdcc,amdsevsnp,tdx}. A confidential virtual machine (CVM), also called a CPU TEE, encrypts memory, restricts address-translation requests, and isolates execution, so computations remain hidden from the cloud provider. GPU TEEs extend the same protection to accelerators, enabling confidential LLM inference to run at scale. The most recent GPUs support GPU TEEs, including NVIDIA's Hopper H100 and H200~\cite{nvidiacc}, and the Blackwell family, such as the B200 and RTX Pro 6000~\cite{nvidia_confidential_computing}.

Major cloud providers, including Microsoft, Meta, and Google, are now deploying production-scale confidential inference stacks on these CPU and GPU TEEs~\cite{meta2025privateprocessing,microsoft_2025_azure,google:2025:gputee}. For example, Meta's "Private Processing" framework~\cite{meta2025privateprocessing} serves LLM-based features, such as summarizing WhatsApp messages~\cite{whatsapp_private_processing}, while keeping sensitive user data hidden from the cloud provider and from Meta itself.
Together, these efforts point toward an end-to-end serving pipeline in which a hardware root of trust guarantees confidentiality. Under this model, \emph{a malicious cloud operator cannot recover the user prompt, intermediate state, or generated response}, so the pipeline achieves end-to-end privacy for large-scale deployments.
\footnote{Denial-of-service attacks fall outside the scope of confidential LLM serving. They degrade availability rather than confidentiality.}

\bheading{Tokenizer: A Demonstrated Source of Leakage.}
CVMs protect memory content and the computations performed on data, but they do not protect against address-pattern leakage. The \underline{tokenizer} is a major component in the LLM serving pipeline that still leaks access-pattern information, which can, in turn, be exploited to extract the plaintext prompts~\cite{hornetztdxray}. In this paper, we focus on the byte-pair encoding (BPE) tokenizer, which many modern decoder-only LLMs use~\cite{gemma_2025,llama3modelcard,abdin2024phi3technicalreporthighly,qwen3technicalreport}.

BPE tokenizers sit at the front of the serving pipeline, converting raw text prompts into a series of short token identifiers (tokenIDs). The tokenizer splits the prompt into words and treats each word as its individual characters, then repeatedly applies the highest-priority rule from a merge table to join adjacent symbols into larger subword tokens; merging stops when no rule applies. Both the merge and vocabulary steps are performed using hash-based dictionary lookups (more details on the tokenizer are provided in Section~\ref{sec:motiv}). The tokenizer builds these tables with a randomized hash map at the start of any inference server deployment. Once built, the mapping and the memory addresses used to look it up stay fixed for the lifetime of the inference. An adversary who records the address trace of a prompt's tokenization process can replay the same deterministic lookups and reconstruct the prompt, even though the vocabulary table and merge table contents remain encrypted in TEE memory (Section~\ref{sec:motiv}). The tokenizer attack differs from token-length attacks~\cite{weiss2024your, zhang2025time}, where an adversary infers prompt characteristics from the prompt's length alone. Address pattern leakage is more severe because it can reconstruct the prompt directly, not just its length.

Prior work TDXRay~\cite{hornetztdxray} demonstrated an end-to-end attack on Intel TDX, reconstructing user prompts with over 90\% similarity to the original text across several model families. A natural defense is to use Oblivious RAM (ORAM), which hides a program's memory access patterns so that the addresses an adversary observes reveal nothing about the data being accessed. TDXRay, however, found that an ORAM-based oblivious-map defense is impractical on the tokenizer path, with 53$\times$ to 369$\times$ overhead. A practical oblivious tokenizer, therefore, remains an open problem. This work, \otro{}, closes this gap by making oblivious tokenizers practical. 

\subsection{\otro: Practical and Oblivious Tokenizer}
\otro{} makes oblivious tokenization practical by leveraging two key observations: (1) The vocabulary and merge tables are both read-only, and hence one can design and optimize ORAMs for read-only data. (2) It is possible to split a long prompt into smaller chunks and allow GPUs to prefill one chunk while the following chunks are being tokenized. Thus, overlapping the tokenization with GPU prefill reduces or even eliminates the critical path latency of tokenization. 

\bheading{Limitation of prior solutions.} Prior oblivious defense on the tokenizer path is far too slow. TDXRay~\cite{hornetztdxray} reports 53$\times$ to 369$\times$ overhead for a naive oblivious map based on tree-based ORAM, and even a tuned one inflates the tokenizer by 19.3$\times$ (Section~\ref{sec:motiv:naive}), with no prefill window large enough to hide such large overhead. The overhead comes from the tree ORAM's access pattern. Tree-based ORAM memory, organized as a tree structure, serves each access by reading a full root-to-leaf path of $O(log(N))$ accesses. Alternatively, square-root ORAM (SqrtORAM)~\cite{oram-orig} alleviates this by generating a single request per access, reducing the number of requests to 3.6$\times$. However, SqrtORAM requires rebuilding the entire memory contents after every $\sqrt{N}$ accesses, thereby creating a bursty $O(N\log^2 N)$ cost. 

\bheading{Rebuild-free SqrtORAM.} \otro{} builds on SqrtORAM and exploits the key insight that all the tokenizer tables are read-only structures to eliminate the rebuild cost of SqrtORAM. First, because there are no writes to these tokenizer tables during inference, \otro{} keeps a pool of independently permuted SqrtORAM replicas. It rebuilds a drained replica off the serving path while the other replicas keep serving. Second, to hide the rebuild cost entirely, we rely on chunked tokenization, which overlaps tokenizing one chunk with the GPU prefill of the prior chunk, so the remaining rebuild cost is hidden by the prefill. The overlap keeps a small, analytically sized replica pool viable at long context lengths, where rebuilding a drained replica would otherwise still stall the LLM execution. Together, these changes enable \otro{} to serve fully oblivious tokenization with near-zero time-to-first-token (TTFT) and time-between-tokens (TBT) overhead for production-scale prompts on Intel TDX hardware with an H100 GPU. The contributions are as follows:

\begin{packeditemize}
    \item \textbf{Read-only SqrtORAM replication.} We show that the tokenizer's read-only tables turn SqrtORAM into a practical ORAM choice. \otro{} serves lookups from a pool of independently permuted replicas, so a drained replica can rebuild off the serving path. It resolves each key through an oblivious addressing layer that hides the access pattern. It also treats the stash size as a tunable parameter and increases the stash size to minimize rebuilds.

    \item \textbf{Chunked tokenization.} We interleave tokenization with GPU prefill over a shared KV cache, so rebuilds overlap with GPU computation. We size the replica pool by balancing rebuild throughput against the rate at which prefill consumes ORAM accesses. The balance determines the replicated pool size, thereby minimizing the memory overhead of replication. We fix the chunk size offline, where prefill saturates.

    \item \textbf{Bounded residual leakage.} We characterize what \otro{} still leaves as observable leakage, and show that the leakage collapses to just the input prompt length without any disclosure of the prompt content. We run a closed-world entropy analysis over real prompts and find that padding drives the ORAM access count down to near zero.

     \item \textbf{End-to-end evaluation.} We implement \otro{} in HuggingFace's Tokenizers library and in vLLM, and evaluate it on Llama-3.1-8B, Qwen3-8B, and Gemma-3-4B/12B inside a TDX CVM with an NVIDIA H100. At production-scale prompts, \otro{} matches plaintext TTFT for Llama, Qwen, and Gemma-12B in the tokenizer and adds no measurable time between token latency (TBT) in the detokenizer.

\end{packeditemize}

\section{Background}
\label{sec:back}

\subsection{Trusted Execution Environment}
Trusted Execution Environment (TEE, or confidential computing) is a hardware-protected \textit{enclave} that executes code and processes data with confidentiality and integrity via memory encryption and isolation ~\cite{kaplan:2016:sevWpaper,intel:2017:tme,intel:2020:tdx}, protecting against privileged software and physical attacks. 

\bheading{Confidential Virtual Machine (CVM).} CVMs extend TEE protection to a virtual machine, so both the guest OS and its applications run inside the hardware-protected boundary~\cite{intel:2020:tdx,amd:2019:sevapi,arm:2021:cca}. Compared to process-scoped enclaves (e.g., Intel SGX~\cite{costan:2016:intel}), a CVM has no fixed enclave size bounded by paging overheads, so it can support multi-TB memory and run data-intensive workloads~\cite{amd:2019:manual,david:2019:sevsnp,kaplan:2017:seves,intel:2020:tdx,intel:2020:tdx_module,arm:2021:cca}. 

\begin{figure}[t]
    \centering
    \begin{subfigure}{0.9\columnwidth}
        \centering
        \includegraphics[width=\linewidth]{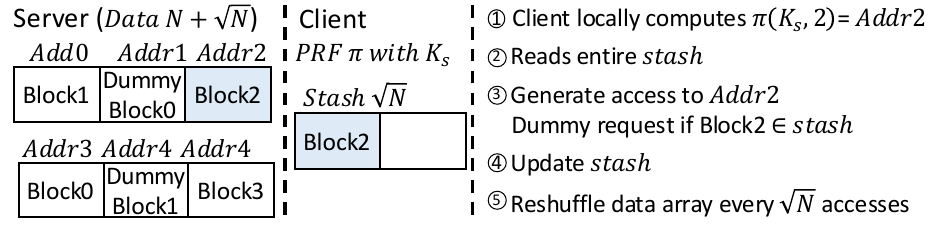}
        \caption{Square-root ORAM ($N=4$).}
        \label{fig:sqrtoram}
    \end{subfigure}
    \begin{subfigure}{0.9\columnwidth}
        \centering
        \includegraphics[width=\linewidth]{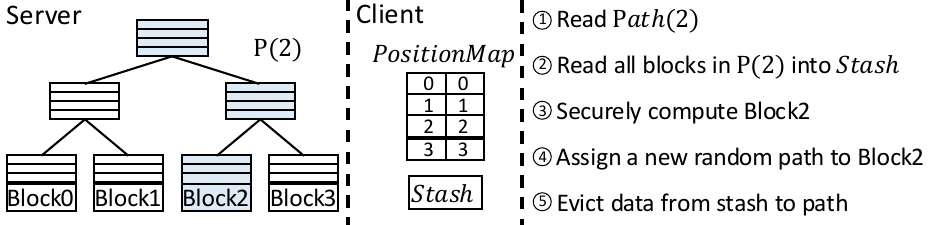}
        \caption{Tree-based ORAM ($N=4$, $Z=4$).}
        \label{fig:treeoram}
    \end{subfigure}
    \caption{SqrtORAM and TreeORAM client-server model.}
    \label{fig:oram}
\end{figure}

\bheading{GPU TEE and Confidential LLM Inference.} NVIDIA’s Hopper (H100) and Blackwell (B200, RTX Pro 6000) GPUs extend TEE protection to the GPU, adding device-rooted attestation~\cite{nvidia:2025:attestation} and per-VM isolated device memory~\cite{nvidiacc,nvidia_confidential_computing}. Together, CPU and GPU TEEs form an end-to-end confidential stack over PCIe, and major providers including Meta, Microsoft, and Google are assembling production confidential LLM inference pipelines on top of this foundation~\cite{meta:2025:whatsapp, whatsapp_private_processing,azure:confidential:ai, mithril:blindai,mithril:blindllama,edgeless:continuum,google:2025:gputee}. Among the options for confidential inference, TEEs are the most practical candidate today, as alternatives like homomorphic encryption or MPC are too slow for interactive serving.
    
\subsection{LLM Tokenization}
\label{sec:bg:tokenziation}
Tokenization, being in the front of the LLM inference pipeline, converts human-readable text into a sequence of discrete token identifiers (tokenIDs). Most widely deployed decoder-only LLMs tokenize with byte-pair encoding~\cite{bpe} rather than  WordPiece~\cite{wordpiece}, and SentencePiece~\cite{sentencepiece}, because BPE handles rare and out-of-vocabulary words by iteratively merging frequent byte pairs into larger tokens. Our work focuses on BPE tokenizers as implemented in HuggingFace’s Tokenizers library~\cite{hftokenizer}, which we describe in more detail in Section~\ref{sec:motivation:case}.

\subsection{Oblivious Random Access Memory}
Oblivious random access memory (ORAM)~\cite{oram-orig} is a cryptographic mechanism that hides memory access patterns from an adversary observing the memory interface. It guarantees that any two logical access sequences of the same length produce computationally indistinguishable physical access traces. Most ORAM constructions assume a client-server model (Figure~\ref{fig:oram}), where a small trusted client keeps the ORAM metadata, and an untrusted server stores the encrypted data blocks. We explain two fundamental ORAM structures, square-root ORAM (SqrtORAM) and tree-based ORAM, using the following notations. $N$ denotes the number of data blocks in ORAM, and $\stash$ is the stash of the ORAM, a trusted cache-like data structure that temporarily stores data on the client. For tree construction, $L=\log N$ is the tree height, and $Z$ is the number of blocks per bucket, counting both real blocks and dummy blocks. $\opath{l}$ denotes the path $l$ .

\bheading{Square-root ORAM.} SqrtORAM stores $N$ data blocks and $\sqrt{N}$ dummy blocks in an encrypted array on the server, and the client holds the pseudo-random permutation, ($\prp$) $\pi$ with key $K_s$ and a $\stash$ of size $\sqrt{N}$. To access a block $Block_a$, the client scans the $\stash$. If $Block_a$ is not in the stash, the client computes $\prp$ $\pi_{K_s}(Block_{a})$ and issues a \textit{real} access request to the main array; if $Block_a$ is already in the $\stash$, the client computes $\prp$ $\pi_{K_s}$ of the next unused dummy and issues a \textit{dummy} access, whiich the server cannot distinguish from a real one. The client then updates the block and reinserts it into the $\stash$. Every $\sqrt{N}$ accesses, an oblivious rebuild reshuffles all $N+\sqrt{N}$ blocks into a freshly permuted array using an oblivious sort. Each access issues one block access to the server, but is preceded by a linear stash scan, so the access cost is $\sqrt{N}$. Additionally, SqrtORAM occasionally incurs a rebuild cost of $O(N\log^2N)$.

\bheading{Tree-based ORAM.} A tree ORAM stores $N$ data blocks into a binary tree of height $L=\log(N)$, and keeps a position map from each block to its current leaf. A position map is a translation table that maps a memory block to a randomly assigned tree path where that data block resides. An access reads the block's entire root-leaf path ($\opath{l}$) into the client's $\stash$, serves and updates the block, remaps the block to a fresh random leaf, and evicts $\stash$ blocks back along the path. Every access transfers a full path and runs eviction, so the per-access bandwidth dominates when accesses are small and frequent, as in tokenizer-table lookups. We use Circuit ORAM~\cite{wang2015circuit} as our tree baseline, the same substrate of TDXRay's mitigation~\cite{oblivious_labs_2024,hornetztdxray}. It shares the same tree layout but differs in eviction. While PathORAM evicts greedily, Circuit ORAM uses a deterministic, data-independent eviction that scans the path while holding only $O(1)$ blocks at a time, which makes the eviction a small oblivious circuit. In our evaluation, we use the more robust circuit ORAM that uses oblivious eviction, providing stronger protection against side-channel leakage.

\section{Leakage from Tokenizer Execution Path}
\label{sec:motiv}

\begin{figure}[t]
    \centering
    \includegraphics[width=0.5\columnwidth]{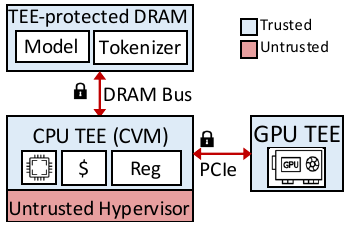}
    \caption{Threat model of a Confidential LLM Service Deployment.}
    \label{fig:threat}
\end{figure}

The tokenizer’s execution path leaks information even when LLM inference runs inside a TEE. We first define the threat model, then present a motivating example to show how input-dependent tokenizer behavior exposes structural access patterns, and we summarize a demonstrated real-world attack. Finally, we measure naive oblivious defenses and show that their overhead does not satisfy the latency requirement of LLM serving.

\subsection{Threat Model}
\label{sec:motiv:threat}
Figure~\ref{fig:threat} shows the system topology of confidential LLM serving. The CPU TEE hosts the CVM, which protects the guest OS, GPU (NVIDIA) driver, and software stacks, including PyTorch and the tokenizer. The GPU TEE stores model weights and activations in on-package HBM. Following existing TEE deployments~\cite{nvidiacc}, we place HBM and the GPU L2 cache in the trusted computing base because they reside within the GPU package, and we assume that all CPU--GPU traffic over PCIe is encrypted and authenticated.

We adopt the standard threat model for CPU and GPU TEEs~\cite{amdcc,google:2025:gputee,nvidia_confidential_computing}, where the adversary controls the full software stack and has physical access to the cloud server. Such an adversary observes the tokenizer's memory access pattern at cache-line granularity through software side channels, such as page-fault-controlled channels and prime-and-probe cache attacks~\cite{li2021cipherleaks,li2022systematic,hornetztdxray}. We also assume the adversary knows the tokenizer table layout. A calibration phase of chosen queries localizes the tokenizer tables even under per-process hash randomization, as TDXRay's token-localization stage shows~\cite{hornetztdxray}, so hiding the layout gives no real protection.

The threat model exposes two leakage channels. First, \textit{cache-line-granularity access-pattern leakage} reveals the tokenizer table's entry that each prompt touches, which enables exact reconstruction (Section~\ref{sec:motivation:case}). Second, the \textit{prompt 's byte length} is directly observable, and the number of table lookups a prompt induces is a near-deterministic function of that length, so the lookup count mostly re-encodes length with a content-dependent residual. Our otro{} approach eliminates the first channel and coarsens the lookup count, making its residual contribution negligible (Section~\ref{sec:idea:access}). Byte length itself remains observable, as it does for any practical ORAM, and we handle it separately with request-length bucketing.

\begin{figure}[t]
    \centering
    \includegraphics[width=0.9\columnwidth]{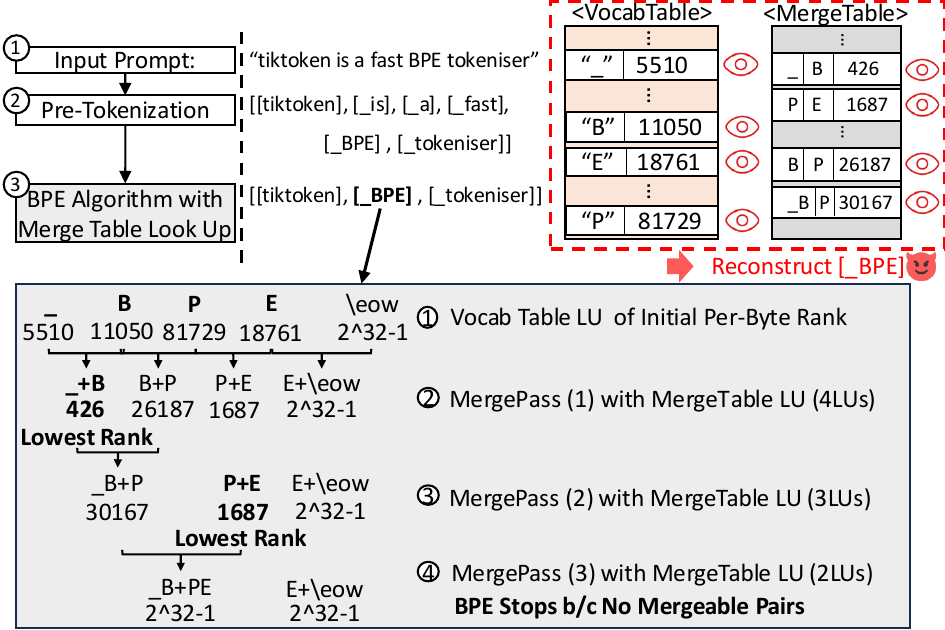}
    \caption{BPE tokenizer architecture and access-pattern leakage in vocab/merge tables. The example produces the tokenIDs for "\_BPE" and shows the accesses an access-pattern-observing adversary sees.}
    \label{fig:tokenizer}
\end{figure}

\subsection{Why BPE Tokenizer Leaks}
\label{sec:motivation:case}

Unlike the branch-free, matrix-multiplication-dominated computation inside transformer blocks, the tokenizer is branch-heavy and table-driven, and it runs directly on the raw prompt. Each character or byte triggers conditional logic, hash lookups, and iterative merge decisions. Because the tokenizer combines \emph{prompt-dependent control flow} and \emph{prompt-dependent memory accesses}, it is a strong side-channel target, since leakage here reveals the user’s original text directly, rather than intermediate embeddings or activations. 

Figure~\ref{fig:tokenizer} walks through the BPE tokenizer on the input \texttt{"tiktoken is a fast BPE tokeniser."} Pre-tokenization first splits the string into per-word pieces, each prefixed with \texttt{"\_"}. The BPE merge algorithm then converts pre-tokenized words into tokenIDs. It first decomposes the word into single-character tokens, each of which is resolved through the vocabulary table. It then repeatedly consults the merge table, which maps an adjacent (id, id) pair to a (Rank, new\_id), and merges the lowest-rank adjacent pair on each merge pass. For \texttt{"\_BPE"}, it merges \texttt{"\_"} with \texttt{"B"} at rank 426, and \texttt{"P"} with \texttt{"E"} at rank 1687, then stops when no mergeable pair remains ($2^{32}-1$ marks an un-mergeable pair). The vocabulary table resolves base characters rather than the whole words, so every word follows this same merge-driven path, \footnote{Some tokenizers expose an \texttt{ignore\_merges} option that emits a word's tokenID directly when the whole word is already a vocabulary entry, skipping the merge loop. Llama-3.1 enables this by default, while Qwen3, Gemma-3, and Phi-3 do not. This option is a prompt-dependent branch that reveals whether a word is a single vocabulary token, so we disable it (\texttt{ignore\_merges}=\texttt{False}) to force every word through the merge loop, yielding a uniform merge-driven path across models and removing the branch. We verified that disabling it produces byte-for-byte identical tokenization across all evaluated models.} whether it ends as a single token (\ie "\_fast") or a sequence of tokens (\ie "\_BPE" $\rightarrow$ \texttt{"\_B"}, \texttt{"PE"}).

The merge procedure is \emph{deterministic}, since the same tables always drive the same merges, and \emph{sequential}, since each merge rewrites its neighboring pairs, so the next pass's candidate set cannot be determined until the current merge completes. Together, these properties allow an adversary to replay the procedure from the access trace alone.

Figure~\ref{fig:tokenizer} shows the access sequence for \texttt{"\_BPE"}. Because the tables are fixed after initialization, the adversary precomputes an address-to-entry map, so each observed address maps back to a concrete entry. The vocabulary-table accesses reveal the base characters ("\_", "B", "P", "E") and the merge-table accesses reveal the mergeable candidate pairs examined at each pass ("\_+B", "B+P", "P+E"). The trace does not name the winning pair directly, but the candidate set of the next pass does. Observing \texttt{"\_B"} being examined in the next pass shows that \texttt{"\_"} and \texttt{"B"} merged, rather than \texttt{"B+P"} or \texttt{"P+E"}. Replaying this deterministic merge sequence reproduces the token sequence, so the address trace alone reconstructs the tokenizer's output without reading any plaintext contents. 

\subsection{A Demonstrated Attack on Intel TDX}
\label{sec:motiv:tdxray}
TDXRay~\cite{hornetztdxray} demonstrates this leakage with a full attack on an unmodified guest CVM running on Intel TDX, recovering user prompts from the tokenizer's memory accesses while all guest memory remains encrypted. In a one-time setup, the adversary builds a host-side tracer that combines TDX page-management APIs (\eg \texttt{TDX.MEM.RANGE.BLOCK/UNBLOCK} with a cache side-channel to pin each access to the 64-byte cache-line it touches. Using custom prompts, it then localizes the tokenizer's table entries in the guest's physical memory. The layout is fixed until the process reinitializes, so this localization cost is amortized over the lifetime of inference. After localization, each victim prompt leaks its full plaintext through its address trace, using the same reconstruction procedure above. TDXRay reports over 90\% reconstruction similarity across model families. 

The tokenizer is therefore a demonstrated channel for prompt reconstruction that survives TEE protections. Defending it requires making the tokenizer's table-access pattern independent of the input, and oblivious memory is the natural tool. We will test its feasibility next.

\subsection{Naive Defenses for Tokenizer Leakage}
\label{sec:motiv:naive}

\begin{figure}[t]
    \centering
    \includegraphics[width=0.9\linewidth]{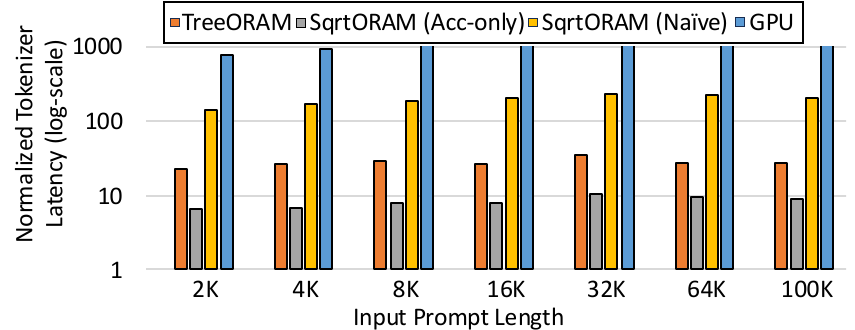}
    \caption{End-to-end tokenizer latency normalized to plaintext baseline (log-scale).}
    \label{fig:latency}
\end{figure}

Oblivious RAM (ORAM) is the natural defense against the address-pattern leakage described in Section \ref{sec:motiv:tdxray}. ORAM guarantees that any two logical access sequences of equal length produce computationally indistinguishable physical traces, so wrapping the tokenizer's vocabulary and merge tables in ORAM breaks the link between observed addresses and tokenIDs. The question is whether ORAM can protect the tokenizer within the strict time-to-first-token (TTFT) budget of interactive LLM serving. 

We integrate naive oblivious defenses into HuggingFace Tokenizers that run the Llama-3.1 tokenizer and measure end-to-end tokenizer latency, normalized to the plaintext baseline (no ORAM protection) (Figure~\ref{fig:latency}). 
\footnote{All ORAM variants in Figure~\ref{fig:latency} are evaluated over a common addressing layer, so the figure isolates the ORAM cost. We analyze the addressing layer in Section \ref{sec:eval}.} We consider three designs: a tree-based ORAM, Square-Root ORAM (SqrtORAM) and a GPU-offloaded tokenizer. The tree-based ORAM is Circuit ORAM~\cite{wang2015circuit}, the ORAM baseline that TDXRay's mitigation builds on~\cite{hornetztdxray,oblivious_labs_2024}, reimplemented in RUST in our stack. 

\emph{GPU offloading is not a viable deployment option.} Offloading tokenization to the GPU, whose device memory and its access pattern the host cannot observe under confidential computing, incurs a 631.7$\times$ slowdown, because BPE's greedy, data-dependent merge sequence causes severe control-flow divergence that defeats SIMD execution.

\emph{Tree-based ORAM is a poor fit for the tokenizer path.} Tree ORAMs organize $N$ entries into a logarithmic-depth hierarchy, and each access traverses a full root-to-leaf path and runs eviction. The tokenizer issues many small lookups, so the path transfer and eviction dominate the cost. The position map makes it worse. It stores each block's current leaf, but it cannot be a flat array, as indexing it with a secret block $b$ would touch an address determined by $b$ itself and leak the very pattern the ORAM hides. Making the position map oblivious requires a recursive construction, so one tokenizer lookup drives one data-tree access plus one access per position-map level, an $O(\log^2{N})$ traversal. Figure~\ref{fig:latency} shows that tree ORAM increases tokenizer latency by 19.3$\times$ over the plaintext baseline. Since tokenization is 3-10\% of plaintext TTFT (Appendix~\ref{sec:appendix:sup}), a 19.3$\times$ tokenizer slowdown does not inflate TTFT proportionally, but it turns tokenization into a first-order term: TTFT reaches a geometric mean of 3.1$\times$ and up to $5.1\times$ for Gemma-4B (Section~\ref{sec:eval:ttft}).

\emph{SqrtORAM has cheap access but a catastrophic rebuild.} SqrtORAM issues a single block access per lookup, which suits the low-latency tokenizer setting. However, every $\sqrt{N}$ access, it obliviously reshuffles all $N+\sqrt{N}$ blocks at $O(N\log^2{N})$ cost. With the rebuild excluded, \emph{SqrtORAM (access-only)} adds only 3.6$\times$ over plaintext, a small fraction of naive SqrtORAM's 114.1$\times$ (Figure~\ref{fig:latency}), and adds only 18\% to TTFT, an overhead a serving pipeline can absorb. SqrtORAM's oblivious access is already cheap enough for the tokenization path. The bursty $O(N\log^2{N})$ rebuild cost is what breaks the TTFT budget. The finding raises \otro{}'s main question: \emph{can we remove the rebuild cost from the critical path entirely while keeping SqrtORAM's cheap oblivious access?} Section \ref{sec:idea} answers this by exploiting the read-only nature of tokenizer tables.
\section{\otro: Practical Oblivious Tokenization Path with Read-Only ORAM}
\label{sec:idea}

\newcommand{\maxctx}{\mathrm{ctx}}          
\newcommand{\toram}{\mathrm{oram}}          
\newcommand{\reshuffle}{\mathrm{rebuild}}   
\newcommand{\access}{\mathrm{acc}}          
\newcommand{\chunk}{c}                       
\newcommand{\att}{\mathrm{att}}             
\newcommand{\prefill}{\mathrm{pf}}
\newcommand{\Lctx}{L_{\maxctx}}             
\newcommand{\tacc}{T_{\access}}             
\newcommand{\treb}{T_{\reshuffle}}          
\newcommand{\sbudget}{s}                     
\newcommand{\Krb}{K}                         
\newcommand{\Pgpu}{P}                         
\newcommand{\pool}{\mathrm{pool}}           
\newcommand{\npool}{n_{\pool}}            

\begin{figure}[!t]
    \centering
    \includegraphics[width=0.9\columnwidth]{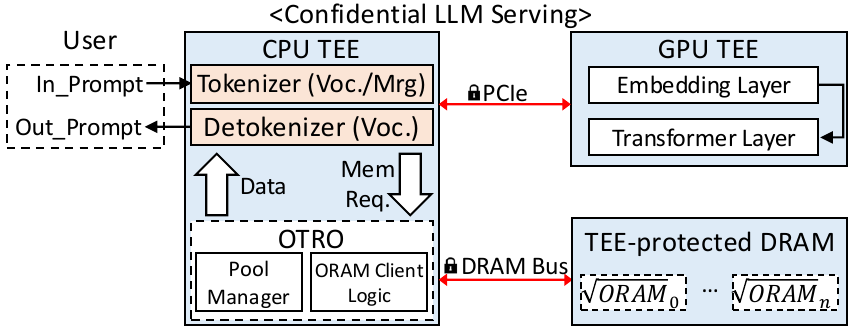}
    \caption{Overview of OTRO system.}
    \label{fig:overview}
\end{figure}

\begin{figure}[!t]
    \centering
    \includegraphics[width=0.9\columnwidth]{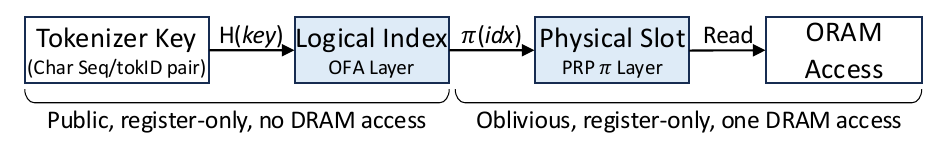}
    \caption{Two-step oblivious lookup. OFA resolves the tokenizer key to a logical index in registers, and the per-replica permutation $\pi$ resolves it to a physical ORAM slot.}
    \label{fig:lookup}
\end{figure}

We present \otro{}, an \emph{efficient and oblivious tokenization path} built on SqrtORAM and tailored tokenizer in LLM serving.

\subsection{\otro{} Overview}
\label{sec:idea:overview}
Section~\ref{sec:motiv:naive} isolated the obstacle. SqrtORAM's oblivious access is cheap on the tokenization path, but its bursty $O(N\log^2 N)$ rebuild, which a single replica must run after a fixed number of accesses, breaks the TTFT budget once it lands on the critical path. \otro{} removes that rebuild from the critical path by exploiting the read-only nature of the tokenizer's vocabulary and merge tables, which stay fixed for the lifetime of inference. A read-only table gives \otro{} two properties that a drop-in SqrtORAM lacks. First, every replica holds identical data, so any replica can serve a lookup while another rebuilds. Second, a rebuild is only a re-permutation of fixed public data, with no writes to synchronize with, so it can run in the background. \otro{} also enlarges the stash beyond the classic $\sqrt{N}$. Obliviousness holds for any stash size, so enlarging the stash is safe and makes rebuilds rare by design. 

\otro{} turns these two properties into two mechanisms. First, \otro{} maintains a \emph{pool of read-only SqrtORAM replicas} and routes each lookup request to a replica that has access budget, while a drained replica rebuilds in the background. This move rebuilds off the serving path, but a drained replica still needs time to rebuild before it can serve again. The second mechanism supplies that time. \emph{Chunked tokenization} interleaves tokenization with GPU prefill, so the rebuild runs underneath GPU work, and the pool can be sized to the access rate rather than to the model's maximum context. The rest of this section covers how the pool serves lookups obliviously (Section~\ref{sec:idea:access}), how chunking and pool sizing keep rebuilds off the critical path (Section~\ref{sec:idea:chunk}), and how the same protection extends to detokenization (Section~\ref{sec:idea:detok}).

Figure~\ref{fig:overview} shows \otro{} under a CVM. Both the tokenizer and detokenizer run in the CPU TEE, while the GPU TEE covers prefill and decoding. A pool manager sits between the tokenizer and the pool, routing each lookup to an available replica and starting asynchronous rebuilds for drained ones. One pool serves all requests reaching the serving instance.

\subsection{SqrtORAM Pool Access}
\label{sec:idea:access}
Tokenization and detokenization only read from the vocabulary and merge tables. Because no access modifies them, a rebuild never has to consider updated data. Every replica contains identical content throughout inference, differing only in physical order. \otro{} therefore maintains a pool of independently permuted SqrtORAM replicas holding the same data, so any replica can serve a request while others rebuild.

\bheading{Resolving a lookup address.} The tokenizer tables are public, so \otro{} does not hide their content. It hides which entries a lookup touches. As Figure~\ref{fig:lookup} shows, a lookup resolves in two steps. First, the addressing layer, Oblivious Fingerprint Addressing (OFA), maps a tokenizer key (a byte sequence for the vocabulary lookup and a tokenID pair for the merge lookup) to a \emph{logical index} using a keyed hash computed entirely in registers. A conventional hash map~\cite{hftokenizer} would issue a key-dependent memory access. \footnote{HuggingFace Tokenizers resolves vocabulary keys through a Rust \texttt{HashMap}, which uses a SwissTable layout. A lookup hashes the key, derives a group index from the hash, and loads that group's control bytes from memory to filter candidates before comparing the full key. Both the control-byte load and subsequent key comparisons are at addresses derived from the key, so the DRAM access sequence is a direct function of the token being looked up. The same holds for any open-addressing or chained hash table, whose probe sequence also starts at a key-derived slot.} OFA also fuses the vocabulary and merge tables into one bucketed structure where each logical index selects a bucket of 16 fixed-size 16B entries. Each entry holds a key field that disambiguates vocabulary from merge lookups and a value field holding the tokenID or (rank, tokenID). Because both tables share one bucket space, a lookup reveals neither which table it hit nor which tokenization phase it belongs to.

Second, the permutation layer maps the logical index to a \emph{physical slot}. Each replica applies its own permutation $\pi$, a keyed Feistel $\prp$ evaluated in registers, so it stores no position array and makes no secret-dependent memory access. The permutation is the source of obliviousness, and it is also why \otro{} avoids the recursive position map that a tree ORAM requires (Section~\ref{sec:motiv:naive}).

\bheading{Initialization.} \otro{} instantiates $\npool$ replicas over the shared bucketed structure. Each replica $i$ applies an independent, per-replica secret permutation $\pi_{K_i}$, so the replicas hold identical data under independently permuted layouts. Because both the tokenizer data and the addressing layer are public, obliviousness rests only on the per-replica permutation's access-pattern indistinguishability, not on layout obscurity. \otro{} keeps all ORAM metadata for each replica, including the per-replica $\prp$ $\pi$, key ($k_i$), $\stash$, and access counter, for each ORAM replica in TEE-protected memory. One pool serves all in-flight requests. 

\bheading{Serving an access.} The pool manager holds a pointer to the active replica, which serves $s$ consecutive accesses. Every request in that window maps to that replica, regardless of the address it requests. On reaching the limit, the manager starts an asynchronous rebuild in a detached thread and advances the pointer to the next replica in round-robin order. Replica selection is thus a deterministic function of the global access count, which any practical ORAM reveals, so the switching encodes nothing about the addresses being accessed. 

The stash size $s$ is a performance parameter, not a security one, since obliviousness holds for any $s$. The reshuffle interval, stash capacity, and dummy reservoir all scale with it. The classic latency-optimal choice is $s=\sqrt{N}$, but \otro{} deliberately uses a larger stash to tradeoff a small increase in per-access scan cost, for far fewer rebuilds. \otro{} therefore selects $s$ per model through offline profiling (Section~\ref{sec:eval:offline}).

\bheading{Pool sizing.} Sized to never drain within one request, a pool would need $\lceil \Lctx \lambda_{\att} / s \rceil$ replicas, where $\Lctx$ is the model's maximum context length, and $\lambda_{\att}$ (equivalently $\tau$) is the measured number of ORAM accesses per token. For long-context models such as Llama-3.1-8B, this reaches hundreds of replicas, which is an impractical memory footprint. A smaller pool works only if drained replicas rebuild while the request is still being served, which needs \emph{slack} that tokenization alone cannot provide. Section~\ref{sec:idea:chunk} obtains that \emph{slack} by overlapping tokenization with GPU prefill and derives the resulting \emph{compact} pool size.

\bheading{Access count padding.} A prompt's lookup count grows with length, so the raw count is itself a coarse leak. \otro{} hides it by rounding the observable count up to a public granularity $B=s/d$, where $d$ is the public divisor. After a prompt's final chunk, the pool manager issues oblivious dummy accesses until the running total reaches the next multiple of $B$. Each dummy access is physically identical to a real one. Padding adds at most $B-1$ dummy accesses per prompt and replaces the exact count with $\lceil access\_count / B \rceil $, which is all that the transcript reveals about the prompt length. A smaller $d$ yields a larger $B$ and coarser, stronger hiding at a higher cost. We analyze the residual leakage in Section~\ref{sec:eval:leak}.

\subsection{Chunked Tokenization} 
\label{sec:idea:chunk}

Section~\ref{sec:idea:access} left a drained replica needing time to rebuild before it can serve again, and \otro{} obtains that time from the GPU. Instead of tokenizing the entire prompt before inference begins, \otro{} partitions the prompt into chunks and interleaves them with GPU prefill, as shown in Figure~\ref{fig:chunk_pipeline}. While the GPU prefills chunk $i$, the CPU tokenizes chunk $i+1$ and any drained replica rebuilds concurrently, so the rebuild runs beneath the prefill work the pipeline already performs, rather than ahead of it. Chunks are prefilled incrementally over a shared KV cache~\cite{agrawal2023sarathi}, so the token sequence and model state match one-pass tokenization exactly.

This interleaving changes the number of replicas \otro{} needs. Without it, the pool must be large enough that no replica drains within a single request, which, for long-context models such as Llama-3.1 (128K tokens), reaches the order of hundreds of replicas. With it, replicas rebuild and return to service while the request is still running, so the pool needs only to sustain the rate at which prefill consumes accesses. 

\begin{figure}[!t]
    \centering
    \includegraphics[width=0.9\columnwidth]{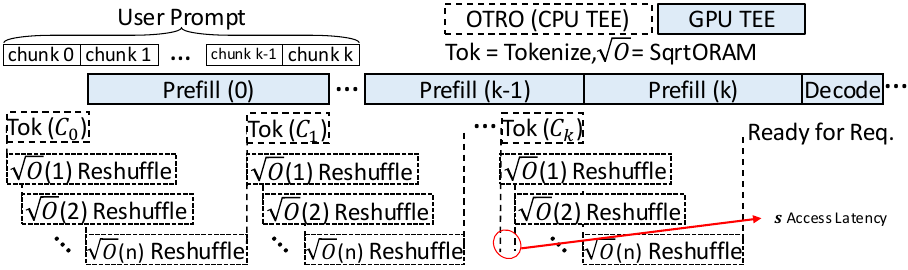}
    \caption{Chunked tokenization with asynchronous reshuffle. Each SqrtORAM replica consumes $s$ accesses before rebuild.}
    \label{fig:chunk_pipeline}
\end{figure}

We use the following notations. Let $\Lctx$ be the model's maximum context length (tokens) and $\Pgpu$ its GPU prefill throughput (tokens/s). Let $\lambda_{\att}$ (equivalently $\tau$) be the number of ORAM accesses per token, a measured per-model constant. Each replica serves $\sbudget$ oblivious lookups before it must reshuffle, holding $N+\sbudget$ physical slots, at a per-lookup latency $\tacc$ that excludes rebuild cost. Rebuilding one drained replica costs $\treb$ under $\Krb$ parallel rebuild threads. The \emph{rebuild throughput} $R=\Krb/\treb$ is the number of replicas rebuilt per second. We provision $\npool$ replicas and split the prompt into chunks of $\chunk$ words. 

\bheading{Effective drain.} Two access rates govern the design. The \emph{sustained drain rate} is the prefill-paced demand $D{=}\tau\,\Pgpu$, the number of accesses the pool must serve per second to keep pace with prefill. Because all in-flight requests share the same GPU, $\Pgpu$ bounds their combined token rate, so $D$ is the pool's aggregate drain regardless of request concurrency. However, when a request tokenizes ahead of prefill, during its burst, the active replica drains at the tokenizer's flat-out rate $1/\tacc$ rather than at the prefill-paced $D$. The pool must survive in both drain cases, so we size it to \emph{effective drain}
\begin{equation}
    D_{\mathrm{eff}} \;=\; \max\!\left(1 / \tacc,\; \tau\Pgpu\right).
    \label{eq:deff}
\end{equation}

\bheading{Pool sizing.} Because chunking overlaps tokenization and rebuild with GPU prefill, the pool must sustain a steady-state balance between draining and rebuilding replicas. This balance holds whenever the rebuild threads can restore access budget at least as fast as tokenization/prefill consumes it, 
\begin{equation}
    R \cdot \sbudget \;\ge\; D 
    \qquad\Longleftrightarrow\qquad
    \frac{\Krb\,\sbudget}{\treb}\;\ge\;\tau\,\Pgpu.
    \label{eq:feasible}
\end{equation}
This is a precondition, not a consequence. If Eq.~\eqref{eq:feasible} fails, no pool size suffices and we must raise $\sbudget$ or the rebuild core count $\Krb$. When it holds, the pool size is the smaller of the max-context bound (enough replicas to never reshuffle within the longest prompt) and the steady-state bound (enough to hide the rebuild behind the prefillprefill),
\begin{equation}
    \npool \,=\; \min\!\left(
        \left\lceil\frac{\Lctx\,\tau}{\sbudget}\right\rceil,\;
        \left\lceil\frac{\treb\,D_{\mathrm{eff}}}{\sbudget}\right\rceil
    \right) 
    \label{eq:pool}
\end{equation}
The compact bound size to whichever stage is the bottleneck. When the GPU is fast, prefill sets the drain, $D_{\mathrm{eff}}{=}D$ (Llama-8B, Qwen-8B). When the tokenizer outpaces the GPU, it runs ahead of prefill and $D_{\mathrm{eff}}=1/\tacc$, which is the case for Gemma-12B with slower prefill. For long-context models, the right term is smaller, replacing an $\Lctx$-scaled pool of over hundreds of replicas with a smaller pool sized only to the drain-to-rebuild ratio, on the order of tens (Section~\ref{sec:eval}). Eq.~\eqref{eq:feasible} assumes $\Krb \ge \npool$, so all in-flight replicas concurrently rebuilds within $\treb$.

\bheading{Chunking decision.} Chunk size is a parameter of the tokenize-prefill schedule, not a per-request quantity. For each model, \otro{} picks a single chunk size offline by profiling prefill throughput on a prompt set and choosing the value that minimizes mean TTFT, then holds it fixed. Larger chunks raise prefill throughput up to a saturation point but also lengthen the exposed first chunk, which suffers ORAM overhead. Thus, the TTFT-minimizing chunk sits or at below that knee. 
We report the profiled values in Section~\ref{sec:eval}.

\subsection{Detokenization}
\label{sec:idea:detok}
Detokenization maps decoded tokenIDs generated by LLM back to text through a reversed vocabulary table. Hiding access patterns to this structure protects the generated response from leakage. It differs from tokenization in two ways. First, each block is padded to the maximum vocabulary string length, so the block size does not leak output token length, giving fewer but larger blocks. Second, decode emits one token every tens of milliseconds, so accesses arrive one at a time, and the rebuild burst never arises, which relaxes the ORAM timing requirement. Both the SqrtORAM pool of Section~\ref{sec:idea:access} and a single Circuit ORAM instance meet this relaxed timing (Section~\ref{sec:eval:tbt}), so we defer the backend choice to the memory evaluation (Section~\ref{sec:eval:system}).
\section{Security Analysis}
\label{sec:security}

\newcommand{\advantage}[2]{\textsf{Adv}^{#1}_{#2}}
\newcommand{\negl}{\textsf{negl}}
\newcommand{\exe}[1]{\textsf{Exec}(#1)}
\newcommand{\addr}[2]{a^{#1}_{#2}}
\newcommand{\trac}[1]{\mathcal{T}(#1)}
\newcommand{\traco}[1]{\mathsf{T}(#1)}

\bheading{Threat model.} We formalize the $\adv$ of Section~\ref{sec:motiv:threat} as a probabilistic polynomial time (PPT) $\adv$ that controls the untrusted software stack. The $\adv$ observes the cache-line-granularity access pattern of every DRAM access the CVM generates, while TEE memory protection hides the contents, cryptographic keys, and CVM internal state. \otro{}'s guarantee is stated against this $\adv$.

\bheading{Security goal and scope.} \otro{} protects the table-dependent memory behavior of tokenization and detokenization. Even observing the full cache-line trace, $\adv$ learns nothing about which entries a lookup touched. The tokenizer path uses the pooled SqrtORAM proven below; the detokenizer path uses a standard Circuit ORAM~\cite{wang2015circuit}, applied unmodified with its recursive oblivious position map, so the leaf lookup issues no address-dependent access and its obliviousness follows directly from that work. The guarantee covers both the key-to-entry resolution, which \otro{} computes in registers, and the table access itself, which the ORAM protects. \otro{} does not make the tokenization computation oblivious. \footnote{The BPE merge loop operates on a small per-word symbol buffer. For typical words, this buffer spans only a constant number of cache lines, so its access pattern reveals at most the word's length and token count, quantities already known through input-length leakage. The plaintext value and the address of the accessed entry remain hidden by OFA and the ORAM, so this residual channel exposes structural information rather than tokens. Fully removing it would require an oblivious BPE implementation with a fixed-capacity buffer and masked updates, which we leave to future work.} Consistent with prior ORAM-based systems ~\cite{pathoram,wang2015circuit}, we do not address computation-level side-channels such as instruction timing. In particular, \otro{}'s in-bucket and overflow key comparison uses a data-dependent branch rather than a conditional move, which is invisible to the cache-line adversary we model but would be a residual under a branch-predictor adversary. The remaining oblivious paths, the register-only PRP and the AVX2 stash scan, are already branchless, so this comparison is the sole residual, and converting it to a conditional move would remove it. \otro{} likewise does not hide application-level structure such as processing time or input length.

\bheading{Observable Transcript.} The observable transcript is the sequence of cache-line-granularity accesses arising from (1) table lookups and (2) background rebuild operations. One pool serves all in-flight requests, so the transcript is their interleaving, which reveals request arrival timing already visible to $\adv$ but nothing about the addresses accessed. Public parameters, including vocabulary/merge table size, pool size $\npool$, stash size $s$, chunk size $\chunk$, and the maximum context $\Lctx$, are revealed by design. With padding enabled, each request's count is rounded to  $\lceil \text{access\_count} / B \rceil$ for a public granularity $B=s/d$, revealing only the coarsened count. Because BPE performs a data-dependent number of merge lookups, the count carries content-dependent residual beyond length, so padding removes that residual. \otro{} may reveal that a request fell into a particular count range and was tokenized across two chunks, but never reveals which entries were accessed or how often. 

\bheading{Proof Strategy.} \otro{}'s obliviousness reduces to that of standard SqrtORAM, which guarantees that for a fixed number of accesses the physical trace is computationally indistinguishable across any two logical access sequences~\cite{oram-orig}. \otro{} makes the following three modifications (M1$\sim$M3), each which we prove to be security-preserving.
\emph{Register-only addressing (M1):} OFA resolves the tokenizer key to a logical index through a keyed hash evaluated entirely in registers, issuing no DRAM access, so the addressing phase emits no cache-line observation. Its key-dependent output is never materialized in memory. Fusing the vocabulary and merge tables into a single bucket space also hides which table an access targets. Also, \otro{} scans each bucket and the overflow stash unconditionally, so occupancy and hit position induce no observable difference.
\emph{Read-only replication (M2):} because the tokenizer tables are read-only, every replica holds identical values, and a rebuild only re-permutes them. Independent keys make the replicas' permutations mutually independent, so replication does not add input-dependent observation beyond a single SqrtORAM.
\emph{Public epoch scheduling (M3).} The replica $j(i) = \left\lfloor \frac{i-1}{s} \right\rfloor \bmod \npool + 1$ and the rebuild timing depend only on the public access index, not on the address requested. Rebuild is an oblivious bitonic re-permutation, whose compare-exchange schedule depends only on the public padded length.
These modifications preserve the access-pattern indistinguishability of the underlying standard SqrtORAM baseline, which we state formally next.

\bheading{Security Proof Sketch.} 
\begin{theorem}[Execution Trace Indistinguishability] 
\label{the:1}
For any two executions inducing the same number of tokenizer accesses, no PPT $\adv$ can distinguish their cache-line-granularity traces with non-negligible advantage.
\end{theorem}
\begin{proof} [Proof Sketch] 
\renewcommand{\qedsymbol}{}
By M1, the addressing phase emits no observation, so the trace consists of the ORAM read and rebuild phase.
By Lemma~\ref{lemma:1} (Appendix~\ref{sec:appendix:security}), the per-access trace distribution is identical across two executions, since replica selection is deterministically from $i$ and public parameters (M2, M3), and the $\prp$ maps the logical address to a location uniform over $N+s$ slots. 
By Lemma~\ref{lemma:2}, traces from distinct accesses are mutually independent, as per-replica $\prp$ are keyed independently and rebuilds follow a fixed public schedule (M3).
The joint distribution, therefore, factors into identical per-access distributions, so the distinguishing advantage is negligible. Full proofs are in Appendix~\ref{sec:appendix:security}.
\end{proof}

\bheading{Detokenization.} Theorem~\ref{the:1} covers the tokenizer path. The detokenizer's Circuit ORAM is oblivious by~\cite{wang2015circuit}, and it is keyed independently of the tokenizer's replicas, so the two produce independent cache-line traces by the same argument as Lemma~\ref{lemma:2}. Observing both, therefore, reveals nothing beyond observing each, and \otro{}'s guarantee extends to the full tokenization and detokenization path.
\section{Evaluation}
\raggedbottom
\label{sec:eval}

\begin{table}[t]
\centering
\footnotesize
\caption{Tokenizer configurations of the evaluated models.}
\begin{tabular}{lccc}
\toprule
\textbf{Tokenizer Config} & \textbf{LlaMA-3.1} & \textbf{Qwen3} & \textbf{Gemma-3}\\ 
\midrule
Vocabulary Size      & 128{,}000 & 151{,}643 & 262{,}144    \\
Merge Rules          & 280{,}147 & 151{,}387 & 514{,}906    \\
Max String Length    & 256       & 256       & 93           \\
Context Length       & 128{,}000 & 32{,}000  & 128{,}000    \\ 
\bottomrule
\end{tabular}
\label{tab:model_config}
\end{table}

\subsection{Experimental Setup and Implementation}
\label{sec:method}
\bheading{System setup.} We run our experiments on an Intel Xeon Gold 6548Y$+$ and an NVIDIA H100 GPU (Driver 550.163.01). The host runs Ubuntu 25.04, and the guest CVM (Intel TDX) runs Ubuntu 24.04 with 64 vCPUs and 300\,GB memory. We implement the ORAM logic inside unmodified TDX. While we inherit the deterministic AES-XTS behavior of TME, a reshuffled block is re-encrypted under a new tweak at every rebuild, so deterministic ciphertext does not let the adversary track blocks across rebuilds. 

\bheading{Implementation and deployment strategy.} We build \otro{} on HuggingFace Tokenizers (v0.22) by replacing the tokenizer's vocabulary and merge lookups with \otro{}'s pooled SqrtORAM. We disable the tokenizer's software merge cache in all configurations, because its hit and miss pattern depends on the prompt and leaks information outside the ORAM-protected path. Chunked tokenization (CT) interleaves tokenization with prefill through an extension to vLLM's~\cite{kwon:2023:vllm} V1 engine with incremental prefill. \otro{} does not change any model weights, CUDA kernels, or the transformer architecture, and it needs no retraining or model reconfiguration. The ORAM-backed tokenizer is agnostic to the serving framework and deploys as a thin layer above its existing prefill and KV-cache paths.

\bheading{Comparisons.}
We compare \otro{} against a plaintext baseline with no obfuscation, and two oblivious designs. \emph{TDXRay} uses cuckoo-hashmap addressing over a tree ORAM backend, and \emph{Tree+OFA} pairs \otro{}'s Oblivious Fingerprint Addressing (OFA) with the same tree ORAM. We report \otro{} at two pool sizes (Section~\ref{sec:eval:offline}): \emph{SqrtPool+OFA (Full)}, sized to the model's maximum context so no rebuild occurs during a request, and \emph{SqrtPool+OFA (Compact)}, sized for steady state where drain rate matches rebuild throughput. At the Compact pool size, Chunked Tokenization (CT) is used to interleave tokenization with prefill to reduce TTFT further.

\bheading{Metrics, models and benchmarks.} 
We measure time-to-first-token (TTFT), time-between-token (TBT), memory footprint, initialization cost, serving throughput, and residual leakage for Llama-3.1-8B~\cite{llama3modelcard}, Qwen3-8B~\cite{qwen3technicalreport}, Gemma3-4B, and Gemma3-12B~\cite{gemma_2025} under PyTorch-2.6. These models span a range of vocabulary sizes and merge-table configurations (Table~\ref{tab:model_config}). Because ORAM overhead depends on tokenizer table size rather than model depth, these models are representative. We sweep prompt length over WikiText prose~\cite{merity2016pointer} at 2K, 4K, 8K, 16K, 32K, 64K, and 100K words, up to each model's 128K-token maximum context. Qwen3's native 32K context is extended to 128K via YaRN~\cite{peng2024yarn}, so every model covers the same range. For end-to-end serving, we measure throughput, TTFT, and TBT at P90 and P99 under load at 4K, 8K, and 16K context (Section~\ref{sec:eval:serving}). For residual leakage, we tokenize 20K ShareGPT prompts and measure the entropy an access-pattern adversary leaves on the prompt (Section~\ref{sec:eval:leak}). 

\begin{table}[t]
    \centering
    \caption{Pool-size derivation. Measured: $\tau, \lambda_w, \Pgpu, R$. Derived $D$ ($1/\tacc $, $\tau\Pgpu/s$) and pool size $\npool$ (Full/Compact). $s{=}16000$ for Llama-8B, Qwen-8B and $20000$ for Gemma models.}
    \footnotesize
    \begin{tabular}{lccccccc}
        \toprule
        \textbf{Model}  & $\boldsymbol{\tau}$   & $\boldsymbol{\lambda_w}$ & $\boldsymbol{\Pgpu}$ & $\boldsymbol{D}$  & $\boldsymbol{R}$ & $\boldsymbol{\npool}$ \\
        \midrule
        Llama-8B      & 15.5  & 17.7  & 25.0K  & 17.5, 24.2  & 28.0  & 111/29    \\
        Qwen-8B       & 14.9  & 17.8  & 24.8K  & 19.2, 24.8  & 30.8  & 112/26    \\
        Gemma-4B      & 20.2  & 23.1  & 47.6K  & 17.0, 48.1  & 13.6  & 115/--    \\
        Gemma-12B     & 20.2  & 23.1  & 16.4K  & 17.0, 16.6  & 13.6  & 115/40    \\
        \bottomrule
    \end{tabular}
    \label{tab:poolsize}
\end{table}

\begin{table}[t]
    \centering
    \caption{Per-token prefill latency (ms/token) vs.\ chunksize, without tokenization. Latency flattens past the saturation.}
    \footnotesize
    \begin{tabular}{r|cccc}
    \toprule
    \textbf{\#Tok} & \textbf{Llama-8B} & \textbf{Qwen-8B} & \textbf{Gemma-4B} & \textbf{Gemma-12B}  \\ \hline
    256  & 0.073 & 0.075 & 0.060 & 0.121\\
    512  & 0.049 & 0.051 & 0.035 & 0.081 \\
    1024 & 0.042 & 0.043 & 0.027 & 0.068 \\
    2048 & 0.040 & 0.041 & 0.022 & 0.061 \\
    3072 & 0.040 & 0.040 & 0.021 & 0.061 \\
    \bottomrule
    \end{tabular}
    \label{tab:prefillsaturation}
\end{table}

\begin{figure*}[t]
    \centering
    \includegraphics[width=1.9\columnwidth]{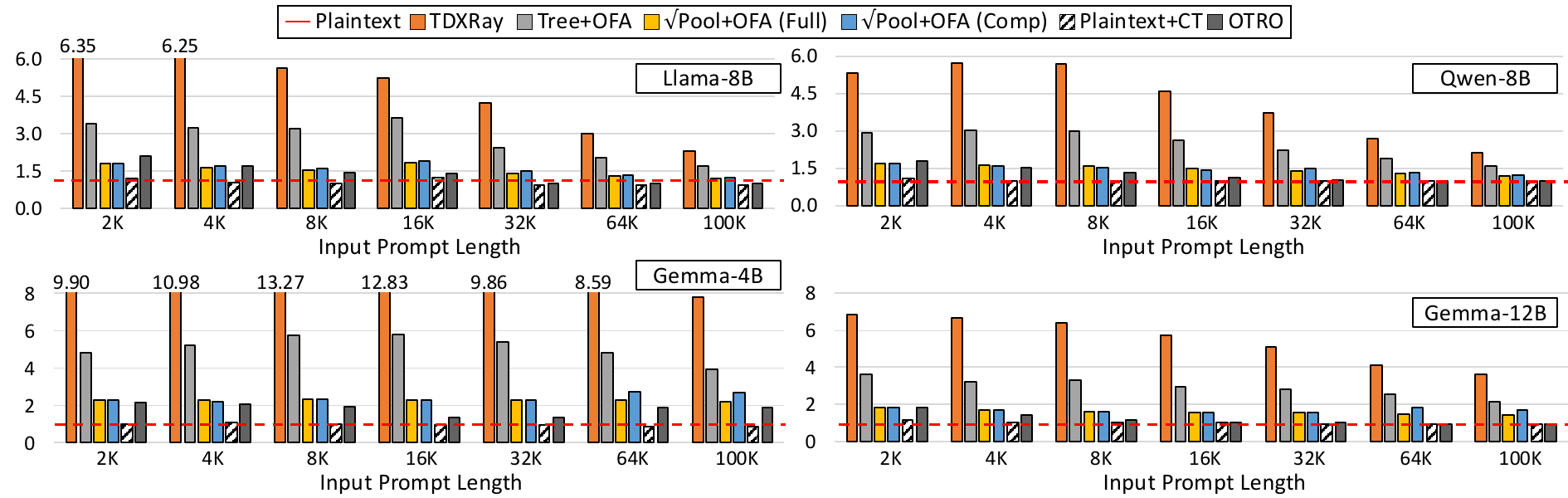}
    \caption{TTFT normalized to Plaintext in Llama-8B, Qwen-8B, Gemma-4B, and Gemma-12B.}
    \label{fig:ttft}
\end{figure*}

\subsection{Offline Tuning for CPU-GPU Overlap}
\label{sec:eval:offline}
\otro{} has two model-specific parameters that we set offline once and hold fixed across all requests. The SqrtORAM \emph{pool size}, $\npool$, governs how rebuilds are hidden and the tokenization \emph{chunk size}, $c$, governs CPU--GPU overlap. Neither is tuned per request. Padding is fixed to $d=10$ (Section~\ref{sec:eval:leak}).

\bheading{Tunable stash size.} We set the stash to be more than the minimal $\sqrt{N}$, which is 639, 550, and 882 for our models, to 16000 for Llama and Qwen and 20000 for Gemma. This stash size is affordable because the per-access scan stays cheap with AVX2 scan at near 2\,$\mu$s which is roughly $20\times$ of the $\sqrt{N}$ cost but still single-digit microseconds. We therefore size the stash by rebuild frequency (Eq~\eqref{eq:feasible}) rather than the latency-optimal $\sqrt{N}$, scaling it with each model's per-token access rate $\tau$ so that rebuilds stay off the serving critical path.

\bheading{Pool size.} Using the same notations from Section~\ref{sec:idea:chunk}, we report two pool decision points. Additionally, let $\lambda_w$ be lookups per word (${=}\tau\cdot\text{tokens/word}$). The \emph{Full} pool holds enough replicas that a maximum-context prompt completes with \emph{no} rebuild. The left term of Equation~\eqref{eq:pool} translates to (total lookups in a max-length prompt) $\div$ (lookups one replica serves). At word length of 100K words this is 111 (Llama, $\lambda_w{=}17.7$), 112 (Qwen, $\lambda_w{=}17.8$), and 116 (Gemma, $\lambda_w{=}23.1$), which we round up to \textbf{120} replicas. 
The \emph{Compact} pool holds only enough replicas to keep serving while drained replicas rebuild. The right term of Eq.~\eqref{eq:pool} sizes it to the effective drain $D_{\mathrm{eff}}=\max(1/\tacc,\ \tau\Pgpu)$, the larger of the tokenizer's serve rate and the prefill rate. A single request drains the pool at prefill's rate on average, but when tokenization outruns prefill it bursts ahead at its own rate, so the pool must cover whichever stage drains it faster. For Llama and Qwen the fast GPU consumes quickly, so the prefill-paced demand $\tau\Pgpu$ is the larger term and sets $D_{\mathrm{eff}}$, giving compact pools of \textbf{29} and \textbf{26}. For Gemma-12B the slower GPU lets tokenization run ahead, so the burst term $1/ \tacc$ satisfies, giving \textbf{40}. All models except for Gemma-4B satisfies the feasibility condition (Eq.~\eqref{eq:feasible}. For Gemma-4B the condition fails, so no compact pool can hide its rebuild, and we instead test Gemma-4B with \textbf{40} pools like Gemma-12B. Table~\ref{tab:poolsize} lists all variables used in the derivation. 

\bheading{Chunking decision.} We fix the chunk at the point that minimizes mean TTFT. To isolate prefill from tokenization and the ORAM, we profile pure prefill cost vs.\ chunk size (pre-tokenized inputs, one decode step, prefix caching off). Table~\ref{tab:prefillsaturation} reports the per-token prefill latency. For every model, the per-token cost falls steeply and then flattens, where this flattening point is the \emph{saturation threshold}. Below it, the GPU is underutilized and above it larger chunks add exposed first-chunk latency without improving throughput. The saturation point scales inversely with per-token compute; the 8B/12B models (Llama, Qwen, Gemma-12B) saturate by 1--2K tokens, while the smaller Gemma-4B saturates later at 3072 tokens. For Gemma-4B, the TTFT optimum sits below this knee, because when past it, the marginal throughput gain in Table~\ref{tab:prefillsaturation} no longer offsets the longer exposed first chunk. We fix each model's chunk, keeping GPU SM utilization within 83--100\%.

\bheading{Detokenizer SqrtORAM config.} Decoding is generally a slow process, and it emits one token every tens of ms, during which the CPU is idle while the GPU runs forward pass. We observed that in our setup a two-replica SqrtORAM pool is therefore sufficient, where each replica's rebuild (CPU-side) fully overlaps the GPU decode kernels and finishes long before its stash is drained. For detokenizer, we compare the SqrtORAM pool with the tree ORAM.

\subsection{Time-To-First-Token Latency} 
\label{sec:eval:ttft}
otro{} makes oblivious tokenization \emph{effectively free in TTFT} for production-scale prompts. Figure~\ref{fig:ttft} reports TTFT from 2K to 100K, normalized to the plaintext tokenizer for Llama-8B, Qwen-8B, Gemma-4B, and Gemma-12B.

\bheading{Cross-model trends.} For Llama-8B, Qwen-8B, and Gemma-12B, overhead shrinks as the prompt amortizes the oblivious tokenizer cost, falling from 1.8--2.1$\times$ at 2K to $-2.4\%$, $-2.4\%$, and $-4.1\%$ at 100K for Llama-8B, Qwen-8B, and Gemma-12B, respectively. \otro{} keeps overhead low in two stages: OFA and the SqrtORAM pool cut the raw oblivious cost to a small residual (Tree+OFA and the pool together, above), and chunked tokenization then hides that residual under GPU prefill. The residual is small enough that the prefill overlap offsets it at the prompt lengths that dominate production serving. The worst case is at 2K, but TTFT there is only $\approx$260\,ms, so the short-prompt gap is negligible and closes as prompts grow into the production-scale.

\bheading{Infeasibility of Gemma-4B.} Gemma-4B is the sole model whose ORAM overhead never decreases, holding at 1.4--1.9$\times$ beyond 16K. This is a fundamental limit of oblivious tokenization on a fast-prefilling model. The 4B model prefills roughly 2--2.9$\times$ faster than the 12B model, so its drain rate outruns the rebuild throughput (Eq~\eqref{eq:feasible}); a large tokenizer makes each rebuild slow, and fast prefill provides little slack to hide it behind. We report Gemma-4B as the limit of the approach.

\bheading{Ablation from TDXray to SqrtPool+OFA.} We add one mechanism at a time: TDXRay~\cite{hornetztdxray} $\to$ Tree+OFA $\to$ SqrtORAM-Pool+OFA (Compact). TDXRay and Tree+OFA share the same tree-based ORAM (Circuit ORAM ~\cite{wang2015circuit}) and recursive position map, and differ only in addressing. TDXRay resolves the address by cuckoo hashing~\cite{oblivious_labs_2024}, which reads both keys' \emph{two} candidate locations because an oblivious lookup cannot reveal which one holds it, so two tree paths per lookup. OFA computes the single physical slot from the token fingerprint in registers and reads \emph{one} tree path per lookup, making Tree+OFA 1.63, 1.67, and 1.83$\times$ faster than TDXRay in geomean for Llama, Qwen, and Gemma-12B, respectively. From tree ORAM to the pooled SqrtORAM (Compact), we replace the tree's $O(\log N)$ path read and recursive position map with an oblivious stash scan (optimized with AVX2). The rebuild is amortized across the pool off the critical path, adding a further 1.73$\times$, 1.65$\times$, and 1.74$\times$.

\bheading{Chunked tokenization.} Chunked tokenization (CT), on top of SqrtPool+OFA(Compact), overlaps each rebuild with vLLM's prefill of the earlier chunks. Chunk size is set per model near the prefill saturation (Section~\ref{sec:eval:offline}). Without it, the compact pool adds 24--71\% TTFT at 100K for the three models. With it, \otro{} matches or beats plaintext by $-2.4\%$ for Llama and Qwen and $-4.1\%$ for Gemma-12B, because the prefill overlap removes the small residual cost that the pool and OFA had already reduced.

\bheading{Compact vs.\ Full pool.} The Compact pool provisions 4.1--4.6$\times$ fewer replicas than Full for Llama/Qwen (29, 26 vs.\ 120) and 3$\times$ for both Gemma models. The two are indistinguishable ($\pm1\%$) up to 16K for Llama/Qwen and 32K for Gemmas. Beyond that, the Compact costs 5--12\% more TTFT because a long prompt exhausts it and forces rebuild onto the critical path. Because chunked tokenization hides this cost, we default to Compact for \otro{} and keep Full as the memory-heavy, latency-optimal alternative.

\bheading{Padding overhead.} All runs above fix the padding policy by $d=10$. The padding cost is largely hidden with chunked tokenization, and we analyze it in Section~\ref{sec:eval:leak}.

\subsection{Detokenization overhead (TBT)}
\label{sec:eval:tbt}
Oblivious detokenization is effectively free, as both the tree ORAM and pooled SqrtORAM backends keep TBT within 1\% of the plaintext across models. Unlike prefill, decode spreads accesses one per token, so the burst that drives the TTFT design (Section~\ref{sec:idea}) never arises. One $\text{tokenID}\to\text{bytes}$ lookup per token takes microseconds, negligible compared to the tens-of-milliseconds decode step. Because both backends meet the TBT budget equally, the detokenizer backend is chosen based on memory constraints rather than latency (Section~\ref{sec:eval:system}).

\begin{table}[!t]
\centering
\footnotesize
\caption{Memory overhead over plaintext tokenizer (MB). SqrtPool shows steady-state and transient footprint.}
\begin{tabular}{lrrr}
    \toprule
                & Llama  & Qwen   & Gemma \\
    \midrule
    \multicolumn{4}{l}{\textbf{Tokenizer (encode)}}       \\ \hline
    TDXRay             & 161          & 123          & 221 \\
    Tree+OFA           & 69           & 56           & 112 \\
    Pool+OFA (Compact) & 1606/1954    & 1213/1618     & 2382/2879 \\
    Pool+OFA (Full)    & 6557/2061    & 4936/1760    & 13840/4557 \\
    \midrule
    \multicolumn{4}{l}{\textbf{Detokenizer (decode)}} \\ \hline
    Tree (Circuit)     & 50           & 56           & 68 \\
    Pool (2)           & 228/524      & 165/460      & 409/963 \\     
\hline
\end{tabular}
\label{tab:mem}
\end{table} 
\subsection{System Level Cost}
\label{sec:eval:system}
\bheading{Memory Footprint.}  Table~\ref{tab:mem} reports memory overhead over the plaintext tokenizer. For the tokenizer, OFA's hash addressing removes the cuckoo tables, reducing Tree+OFA's memory 2.2$\times$ versus TDXRay. The pooled SqrtORAM is the dominant memory cost. Full holds 4.9--13.8\,GB at steady-state, and Compact reduces this by 4.6$\times$ to 1.2--2.4\,GB.
For the detokenizer, both backends tie on TBT (Section~\ref{sec:eval:tbt}), so memory constraints is the deciding factor. Tree ORAM costs only 50--68\,MB with no rebuild-induced transient footprint, a 3.0--6.0$\times$ smaller than the two-replica pool. Thus, \otro{} uses tree ORAM for the detokenizer.

\bheading{Initialization.} Initialization is a one-time cost at server setup and amortized over the serving lifetime, so it never affects a request. Building the oblivious tokenizer raises this cost from 0.7--2.2\,s for plaintext to 2.0--4.7\,s in SqrtPool+OFA(Compact), which we use for \otro{}.

\begin{figure}[t]
    \centering
    \includegraphics[width=0.9\columnwidth]{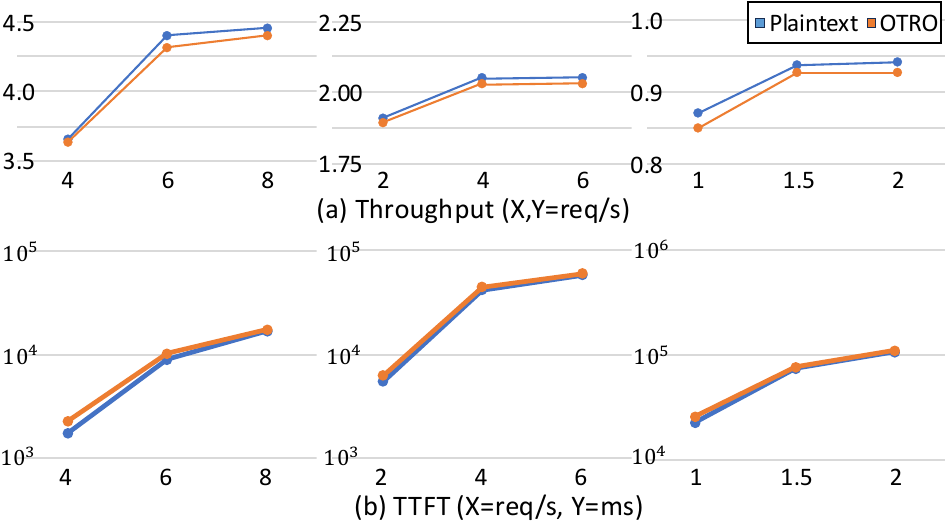}
    \caption{End-to-end serving of Llama-3.1-8B at 4K, 8K, and 16K context. (a) Achieved throughput and (b) TTFT P99.}
    \label{fig:serving}
\end{figure}

\subsection{End-to-End Serving.}
\label{sec:eval:serving}
We evaluate \otro{} in a full serving loop under realistic load, rather than on isolated lookups (Figure~\ref{fig:serving}). The baseline is the unmodified plaintext stack (plaintext tokenizer/detokenizer, no chunked tokenization). \otro{} adds the oblivious tokenizer/detokenizer and chunked tokenization, so the comparison measures the end-to-end cost of an oblivious tokenization path. We run Llama-8B at 4K, 8K, and 16K context, sweeping the offered load from moderate up to the throughput plateau. The chunk size is the same as Figure~\ref{fig:ttft} with 1024 words per chunk. We report throughput, TTFT P99, and TBT here, and defer TTFT P99 and TBT tables and run configurations to Appendix~\ref{sec:appendix:sup}.

\otro{} matches plaintext throughput at every measured load, staying within 2\% of plateau (4.40 vs.\ 4.45 req/s at 4K, 2.07 vs.\ 2.11 at 8K, 0.93 vs.\ 0.94 at 16K), because the pool sustains the access rate that prefill consumes and never pays for the rebuild. TTFT P99 is within 3.3\%--4.4\% at the plateau and rises to 27.7\%, 14.5\%, and 13.0\% at lower loads at 4K, 8K, and 16K, respectively. The lower-load overhead is the exposed first-chunk tokenization, a fixed cost \otro{} cannot overlap. It vanishes as prompts lengthen or load rises, which is when TTFT matters. At idle loads, the absolute gap stays small, 490 ms at 4K. TBT is mostly unchanged, within about 2\% of plaintext at P90 and P99 across all loads (Appendix~\ref{sec:appendix:sup}). \otro{} uses tree ORAM for detokenization (Section~\ref{sec:eval:tbt}).

\begin{table}[!t]
    \centering
    \caption{Residual leakage under padding (20K ShareGPT, Llama-3.1, $s=16000$). The adversary observes (byte-length, padded access count). Higher $H$ and lower $I$ = more privacy.}
    \footnotesize
    \begin{tabular}{l|c|c}
    \hline
    Configuration       & Leaked bits (H) & I \\
    \hline
    No ORAM             & 14.29 (0)                 & ---       \\
    \otro{} (w/o pad)   & 12.69 (1.60)              & 3.94 bits \\
    \otro{} ($d=20$)    & \phantom{0}8.79 (5.5)     & 0.04 bits \\
    \otro{} ($d=10$)    & \phantom{0}8.76 (5.53)    & 0.01 bits \\
    $\ell$ only         & \phantom{0}8.75 (5.54)    & 0.00 bits \\
    \hline
    \end{tabular}
    \label{tab:leakage}
\end{table}

\subsection{Residual Leakage Quantification}
\label{sec:eval:leak}
\bheading{Modeling.} \otro{} hides which table entries a lookup touches, but the \emph{number} of accesses a prompt creates stays observable and correlates with the prompt. We quantify this under a closed-world model: 20K  ShareGPT prompts \cite{sharegpt_unfiltered}, tokenized with Llama-3.1's BPE tokenizer ($s=16000$) with prompts selected uniformly, so the prior uncertainty is $H(\text{prompt}) = \log_2(20{,}000) = 14.29$ bits. Real prompt spaces are far larger, so our analysis provides a \textbf{conservative upper bound} on adversary advantage.

\bheading{Observables.} After ORAM obfuscation, the adversary sees only two quantities: prompt byte length $\ell$ and total ORAM access count. Because \otro{} fuses the vocabulary and merge tables into a single structure, accesses are physically identical, so per-word merge counts or vocabulary-versus-merge lookups are hidden. Only the aggregate count is visible. Epoch and chunk counts are deterministic functions of the access count and the public parameters (\ie $s$, $c$), thereby providing no independent information. The access count itself is a near-deterministic function of byte length ($\text{access\_count} = 3.20 \times \text{byte\_length} + 11.26$, $R^2=0.9915$), so it primarily re-encodes length with a content-dependent residual.

\bheading{Entropy analysis.} Observing byte length alone resolves 8.75\,bits of 14.29\,bits, leaving $H(\text{prompt | byte\_length}) = 5.54$ bits. The access count resolves a 3.94 bits ($I(\text{access\_count}; \text{prompt} | \ell) = 3.94 $ bits), leaving 1.60\,bits, roughly three candidates. Although length explains 99\% of the count's variance $R^2=0.9915$, this residual is most of the uncertainty that remains after length. Byte length accounts for most of the leak, which any practical ORAM incurs unless every request is padded to the maximum context length. The access count is the one channel \otro{} controls, and padding (Section~\ref{sec:idea:access}) drives its 3.94\,bits towards zero. The residual is also length-dependent. For prompts over 500 bytes, $H(\text{prompt|}\ell)$ is below 1.10 bits, so the identification is driven by length, not from \otro{}.

\begin{table}[t]
    \centering
    \caption{Padding cost by $d$ (Llama-3.1, $s$=16000). At most $s/d-1$ dummy accesses are added per request for padding.}
    \footnotesize
    \begin{tabular}{c|c|c|c}
    \hline
    d   & pad (acc.)    & lat.\,@29 (ms)    & lat.\,@120 (ms) \\\hline
    1   & 15999         & 915.8             & 879.9 \\
    2   & 7999          & 457.9             & 439.9 \\
    5   & 3199          & 183.1             & 175.9 \\
    10  & 1599          & 91.5              & 87.9  \\
    20  & 799           & 45.8              & 43.9  \\ \hline
    \end{tabular}
    \label{tab:padcost}
\end{table}

\bheading{Padding.} \otro{} rounds the observed count up to the next multiple of $B=s/d$, so the adversary learns only $\lceil \text{access\_count} / B \rceil$ at a worst-case cost of $s/d-1$ dummy accesses per request. $d=1$ reveals only the epoch number but pads up to $s-1$ accesses (15,999 for Llama-3.1), while larger $d$ cuts the padding by a factor of $d$ at finer granularity. At $d=10$, the count's contribution falls from 3.94 to 0.01 bits, within 0.01 of the floor, increasing residual uncertainty to 5.53 bits and the candidate set 15\,$\times$ (Table~\ref{tab:leakage}). Total leakage is nearly flat across $d$ (8.75--8.79 leaked bits, Table~\ref{tab:padcost}), so coarser padding buys almost no additional privacy at a higher cost. $d=10$ matches $d{=}1$ within 0.01\,bits at 10$\times$ lesser cost (91.5 vs.\ 915.8\,ms). Under chunked tokenization, this padding runs beneath prefill, so its realized TTFT impact is far smaller.

\bheading{Versus plaintext.} Without ORAM the adversary sees the full vocabulary and merge index sequence, which reconstructs the prompt (Section~\ref{sec:motiv}) and leaves $H(\text{prompt|trace})\approx0$. \otro{} replaces this per-token pattern with coarse, length-correlated metadata, raising residual uncertainty to 1.60\,bits unpadded and 5.53\,bits padded. The remaining risk for long prompts is inherent in byte-length leakage, which every practical ORAM inherit, and is not introduced by \otro{}.

\section{Related Work}

\bheading{Side-channel Attacks in CVM. } Side-channel attacks are stronger against TEE than non-TEE settings because the attacker holds privileged interfaces like performance counters~\cite{gotzfried2017cache}, power analysis~\cite{lipp2021platypus}, and controls internal resources~\cite{brasser2017software,dall2018cachequote,moghimi2017cachezoom,schwarz2017malware,van2018nemesis}. Ciphertext attacks against AMD SEV go further and read secrets from encrypted memory traffic~\cite{li2022systematic,li2021cipherleaks}. \otro{} addresses the access-pattern leakage in this threat model.

\bheading{Side channels in LLM inference.} Recent attacks leak tokens through the embedding table. Spill the Beans~\cite{adiletta2025spill} and Gao et al.~\cite{gao2025know} flush and reload embedding rows to recover input and generated tokens from the cache state. These attacks read tokenIDs after tokenization, whereas \otro{} protects the tokenizer, which handles the raw prompt.

\bheading{ML Embedding Table.} Prior work applies ORAM to embedding-table lookups~\cite{laoram,hashemi2022data} or replaces the table with deep hash embedding~\cite{umar2025efficient}. While these work on tokenIDs after tokenization, \otro{} targets the tokenizer.

\bheading{PIR and Oblivious Map.} PIR ~\cite{chor1998private,kushilevitz1997replication} hides the queried index at the protocol layer but does not provide memory-oblivious execution under a DRAM-trace adversary. Oblivious maps~\cite{mishra2018oblix,tinoco2023enigmap} hide the accessed key, but their security still reduces to the underlying ORAM. \otro{} works at the memory-access layer inside a CVM.

\bheading{TEE-based ORAM.} ZeroTrace~\cite{zerotrace}, Obliviate~\cite{ahmad2018obliviate}, and  Oblix\\ ~\cite{mishra2018oblix} run on Intel SGX, with memory limit: they work in a different regime than VM-based TEEs. PRO-ORAM~\cite{tople2018pro} also exploits read-only data and pipelined rebuilds but targets large-file access on SGX storage rather than latency-critical tokenizer lookups in CVM. On VM-based TEEs, Menhir~\cite{reichert2024menhir} runs ORAM logic with Intel TME, while Obelix~\cite{obelix} and MC-ORAM~\cite{wang2026mc} defeat the ciphertext determinism of AES-XTS. Opal~\cite{kaviani2026opal} applies ORAM to oblivious agentic memory. All are complementary to \otro{}, which instead targets fine-grained accesses to the tokenizer table.

\section{Conclusion}
Tokenizers inside TEEs remain an exploitable source of side-channel. Greedy BPE tokenization makes deterministic, data-dependent table accesses, and the resulting DRAM address traces let a privileged adversary reconstruct the user prompt, as recent work showed end-to-end on Intel TDX. otro{} closes this channel efficiently by exploiting that the vocabulary and merge tables are read-only during inference. Because the tables never change, \otro{} keeps a pool of independently permuted SqrtORAM replicas, rebuilds a drained replica off the serving path while the others serve, and enlarges the stash to make rebuilds rare. Chunked tokenization then overlaps each rebuild with GPU prefill, which sizes the pool to the rebuild rate rather than to the maximum context. Together, these turn SqrtORAM's bursty rebuild into background work that never stalls serving. On Llama-3.1-8B, Qwen3-8B, and Gemma-3-12B, \otro{} matches plaintext TTFT at production-scale prompts while reducing the observable trace to prompt length alone, showing that workload-aware ORAM integration is a practical path to confidential LLM serving.

\bibliographystyle{plainurl}
\bibliography{references}

\cleardoublepage

\appendix 
\label{sec:appendix}
\section{Ethical Consideration}
This work is a defensive work. \otro{} mitigates a prompt-reconstruction side channel in confidential LLM serving, and it introduces no new attack or vulnerability. The leakage we address was disclosed by prior work (TDXRay~\cite{hornetztdxray}) against a threat model established for CPU and GPU TEEs~\cite{amdcc,google:2025:gputee,nvidia_confidential_computing}. 

We used only public datasets, WikiText~\cite{merity2016pointer} and ShareGPT~\cite{sharegpt_unfiltered} and no proprietary data.

\otro{} closes the tokenizer side-channel within a practical latency budget in the confidential LLM serving systems. With this closure, \otro{} makes the confidentiality guarantee in these systems already advertise achievable in practice. 

\section{Supplementary Evaluation}
\label{sec:appendix:sup}
\begin{table}[t]
  \centering
  \footnotesize
  \caption{Plaintext tokenizer latency (ms) and its share of TTFT (Llama-3.1).}
  \label{tab:tok-share}
  \setlength{\tabcolsep}{4pt}
  \begin{tabular}{lccccccc}
    \toprule
                 & 2K   & 4K    & 8K    & 16K   & 32K    & 64K    & 100K   \\
    \midrule
    Latency (ms) & 9.41 & 18.78 & 35.59 & 74.11 & 139.97 & 334.13 & 617.23 \\
    \% of TTFT   & 8.78 & 10.02 & 8.97  & 7.63  & 5.65   & 4.60   & 4.07   \\
    \bottomrule
  \end{tabular}
\end{table}

\subsection{Tokenizer Share in Plaintext TTFT}
Table~\ref{tab:tok-share} reports the plaintext tokenizer latency for Llama-3.1 and its share of TTFT across prompt lengths. Tokenization is a small and shrinking fraction of TTFT, from about 10\% at 4K down to 4\% at 100K and roughly 7\% on average. The 19.3$\times$ tokenizer slowdown inflates TTFT far less than proportionally (Sections~\ref{sec:motiv:naive}).

\subsection{Plaintext TTFT Baseline}
The TTFT figures in Section~\ref{sec:eval:ttft} are normalized to the plaintext tokenization, which keeps the comparison across models and prompt lengths on one scale. Table~\ref{tab:pttft} gives the absolute plaintext TTFT in milliseconds behind that normalization, from 2K to 100K words for all four models. These are the denominators for every normalized number in the body.

Two patterns are worth noting. TTFT grows super-linearly with prompt length approaching quadratic at long prompts, because prefill attention is quadratic in sequence length. Gemma-4B is the exception, staying close to linear because its prefill is fast, which is also why it is the hardest case for \otro{} to hide oblivious cost behind (Section~\ref{sec:eval:ttft}).

\begin{table}[t]
    \centering
    \footnotesize
    \caption{Plaintext TTFT (ms) across prompt lengths. }
    \begin{tabular}{lrrrr}
        \toprule
        Length & Llama-8B & Qwen-8B & Gemma-4B & Gemma-12B \\
        \midrule
        2K   &   125.5 &   147.4 &  107.8 &   167.8 \\
        4K   &   240.6 &   261.6 &  187.2 &   335.2 \\
        8K   &   487.0 &   517.0 &  300.0 &   652.3 \\
        16K  &   865.9 &  1244.2 &  596.8 &  1399.6 \\
        32K  &  2709.0 &  3079.3 & 1196.0 &  2901.2 \\
        64K  &  7703.8 &  8948.2 & 2613.2 &  6454.0 \\
        100K & 15849.9 & 18535.3 & 4442.2 & 11076.8 \\
        \bottomrule
    \end{tabular}
    \label{tab:pttft}
\end{table}

\subsection{Initialization Cost}
Initialization is a one-time cost paid once at server startup and amortized over the serving lifetime, so it never affects a request. We report it here rather than in the body because it does not enter the TTFT or throughput results. Table~\ref{tab:init} gives the cost for the plaintext tokenizer and each oblivious configuration. Building the oblivious tokenizer raises startup from 0.6--2.2\,s in plaintext to 2.0--4.7\,s for the Compact pool that \otro{} deploys. The Full pool is higher, 11--19\,s, because it permutes many more replicas, which is one more reason \otro{} defaults to Compact. Even the Full pool's cost is paid once and does not repeat.

\begin{table}[t]
\centering
\footnotesize
\caption{Initialization overhead (ms), a one time startup cost.}
\label{tab:init}
\begin{tabular}{lrrr}
    \toprule
        & Llama & Qwen & Gemma \\
    \midrule
    Plaintext           & 744   & 636   & 2213 \\ \hline
    \multicolumn{4}{l}{\emph{Tokenizer (encode)}} \\ \hline
    TDXRay              & 2763  & 1956  & 4617 \\
    Tree+OFA            & 992   & 919   & 2415 \\
    Pool+OFA (Full)     & 14295 & 10735 & 19076 \\
    Pool+OFA (Compact)  & 4658  & 2286  & 1996 \\ \hline
    \multicolumn{4}{l}{\emph{Detokenizer (decode)}} \\ \hline
    Tree (Circuit)      & 1216  & 790   & 355 \\
    Pool (2)            & 362   & 273   & 543 \\
    \bottomrule
\end{tabular}
\end{table}

\begin{table*}[t]
\centering
\footnotesize
\caption{End-to-end serving of Llama-3.1-8B: plaintext baseline versus \otro{}, across three
contexts and three offered loads each, up to the throughput plateau.}
\label{tab:serving-full}
\begin{tabular}{ll ccc ccc ccc}
\toprule
& & \multicolumn{3}{c}{4K} & \multicolumn{3}{c}{8K} & \multicolumn{3}{c}{16K} \\
\cmidrule(lr){3-5}\cmidrule(lr){6-8}\cmidrule(lr){9-11}
Metric & System & \multicolumn{9}{c}{Offered load (req/s)} \\
 & & 4 & 6 & 8 & 2 & 4 & 6 & 1 & 1.5 & 2 \\
\midrule
\multirow{2}{*}{Throughput (req/s)}
 & Baseline & 3.656 & 4.405 & 4.450 & 1.820 & 2.104 & 2.105 & 0.870 & 0.938 & 0.942 \\
 & \otro{}  & 3.637 & 4.313 & 4.397 & 1.795 & 2.063 & 2.067 & 0.851 & 0.927 & 0.927 \\
\midrule
\multirow{2}{*}{TTFT P90 (ms)}
 & Baseline & 1239.0 & 7393.6 & 14793.6 & 3544.6 & 34929.6 & 51760.6 & 16167.7 & 62285.0 & 93515.3 \\
 & \otro{}  & 1589.3 & 8544.1 & 15445.4 & 4073.0 & 37925.1 & 53627.4 & 18937.4 & 64803.8 & 96722.3 \\
\midrule
\multirow{2}{*}{TTFT P99 (ms)}
 & Baseline & 1781.3 & 9245.1 & 17152.5 & 5410.8 & 41339.2 & 58193.0 & 22873.0 & 73865.0 & 106744.5 \\
 & \otro{}  & 2273.9 & 10550.4 & 17902.0 & 6196.4 & 43534.8 & 60159.9 & 25840.6 & 76345.9 & 110223.9 \\
\midrule
\multirow{2}{*}{TBT P90 (ms)}
 & Baseline & 91.09 & 94.05 & 95.06 & 100.25 & 101.19 & 101.16 & 114.09 & 114.75 & 114.32 \\
 & \otro{}  & 92.87 & 96.17 & 95.48 & 102.37 & 102.96 & 102.69 & 115.92 & 115.78 & 115.97 \\
\midrule
\multirow{2}{*}{TBT P99 (ms)}
 & Baseline & 93.49 & 94.25 & 95.49 & 100.58 & 101.35 & 101.35 & 117.10 & 115.91 & 114.67 \\
 & \otro{}  & 94.82 & 96.34 & 95.70 & 102.72 & 103.11 & 102.92 & 116.52 & 116.29 & 116.36 \\
\bottomrule
\end{tabular}
\end{table*}

\subsection{End-to-End Serving}
\label{sec:}
We compare \otro{} with plaintext baseline without ORAM obfuscation
\otro{}'s tokenizer serves from a pool of 29 replicas at stash size $s{=}16000$, same configuration used in Section~\ref{sec:eval:ttft}. 
We use the tree-ORAM-based detokenizer (Section~\ref{sec:eval:system}). 

We validate the serving path on Llama-8B. Requests run through vLLM's V1 AsyncLLM engine with continuous batching. \otro{} tokenizes each request incrementally along the chunked tokenization, and the plaintext baselines tokenize the whole prompt up front. 
Eight tokenizer workers handle in-flight tokenization, and set the in-flight request cap above every rate we test (64 for 4K, 32 for 8K and 16 for 16K prompt length). Load arrives as an open-loop Poisson stream whose rate we sweep for each prompt length, using WikiText-103 prompts of 4K, 8K, and 16K words.

Throughput is the direct test of whether rebuild stays off the critical path, and \otro{} matches plaintext at every measured load. At the plateau it is within 2\% of plaintext ($-1.2\%$, $-1.8\%$, and $-1.6\%$ at 4K, 8K, and 16K), and the two systems plateau at the same load. Because the pool sustains the access rate prefill consumes, serving throughput does not pay for the rebuild.

TTFT tracks plaintext closely near the plateau and diverges at low load. At the highest measured load per context, \otro{} adds 3.3--4.4\% at P99. At the lowest load it adds 27.7\%, 14.5\%, and 13.0\% at 4K, 8K, and 16K. The gap reflects what dominates TTFT: at low load it is a request's own tokenize-and-prefill work, so the oblivious access cost is visible, while near the plateau it is shared queueing, which is far larger and equal for both systems. The absolute gap stays small, about 490 ms at 4K. P90 follows the same pattern as P99, so the divergence is not driven by a few tail requests.

TBT is unchanged. Across all contexts and loads, \otro{} stays within about 2\% of plaintext at both P90 and P99 (for example, 103.1 vs.\ 101.4 ms at 8K, P99). Detokenization runs on the decode path, but each token's oblivious lookup costs far less than the tens-of-milliseconds decode step, and the detokenizer's rebuild hides under decoding, so the per-token latency the user sees does not move.

\bheading{Tree-ORAM Tokenizer on Serving.}
We used TDXRay's oblivious map design tokenizer with tree-ORAM-based detokenizer running under the same configuration mentioned above. Across offered rates of 4, 6, and 8 req/s at a 4K prompt, TDXRay's throughput stays pinned near 0.85 req/s (0.83, 0.85, 0.90) with a non-binding in-flight cap of 64, so added load does not raise completions. This is fundamental to the omap design rather than an implementation artifact: a single Circuit-ORAM mutates its shared position map, stash, and tree on every access, so all accesses must serialize through one instance and concurrent requests cannot overlap. Parallelism would require replicating the structure, which is exactly what OTRO's read-only SqrtORAM pool does.

\subsection{OFA Collision Test}
OFA addressing is collision-free by construction. Correctness of the merged bucket layout (vocabulary and merge tables are merged Section~\ref{sec:idea:access}) requires that no two distinct keys share both a bucket and a keyfield, since the in-bucket scan would otherwise resolve to an ambiguous value. The builder enforces this directly. It draws a hash seed, checks all keys for such a collision, and redraws until none remains, aborting rather than proceeding if it fails. Therefore it either emits a collision-free table or fails at load. Merge entries match on their exact eight-byte identifier pair and never collide with each other, and merge and vocabulary entries are separated by a type bit, so the only collisions the search can encounter are between two vocabulary fingerprints. We design a collision experiment where we swept the seed space over the real vocabularies of Llama-3.1, Qwen3, and Gemma-3 over 64-redraw of the hash seed. For every model, the very first seed was already collision-free, and across 10,000 random seeds not one produced a collision, placing the per-seed collision probability below 3e-4 and the probability of exhausting all 64 redraws below 3.4e-226. The same builds produced no bucket overflow, so every lookup resolves within a single fixed bucket read.

\section{Security Analysis}
\label{sec:appendix:security}
Here we provide the proofs for two lemmas used in the proof sketch of Theorem~\ref{the:1} in Section~\ref{sec:security}.

\begin{lemma}[Per-Access Trace Indistinguishability]
\label{lemma:1}
For all $i \in [N]$, the random variables $\trac{\addr{1}{i}} \ \text{and}\ \trac{\addr{2}{i}}$, denoting the cache-line-granularity access traces induced by the $i$-th logical access in executions $\exe{1}$ and $\exe{2}$, are identically distributed, i.e.,
\[
\trac{\addr{1}{i}} \;\equiv_d\; \trac{\addr{2}{i}} .
\]

\begin{proof}
    We analyze the cache-line-granularity access trace generated by the $i$-th logical access and show that it is identically distributed for $\exe{1}$ or $\exe{2}$. We decompose the trace into two components: 1) accesses incurred during the ORAM read operation, and 2) accesses incurred during reshuffling.

    \bheading{Read phase.} For the ORAM read phase, the controller deterministically selects the SqrtORAM replica to access based solely on the logical access index $i$. Specifically, the $j$-th SqrtORAM replica is selected, where
    \[
        j = \left\lfloor \frac{i-1}{s} \right\rfloor \bmod n_{pool} +1
    \]
    and $n_{pool}$ denotes the number of replicated SqrtORAM replica, and $s$ denotes the stash size. Since this selection depends only on $i$ and public system parameters, the same SqrtORAM replica is accessed in executions $\exe{1}$ and $\exe{2}$.
    Within the selected SqrtORAM replica, the physical location is determined by the replica's keyed permutation over its $N+s$ slots. Under the pseudorandom permutation ($\prp$), this placement is computationally indistinguishable from that of a uniform random permutation~\cite{oram-orig}, so far for any physical location $\ell$ in the SqrtORAM, the probability that $\ell$ is accessed during the read phase is
    \begin{align}
        \Pr[\text{access } \ell] = \frac{1}{N + s} \nonumber.
    \end{align}
    In addition, the controller performs a linear scan of the stash during every read operation. This scan accesses the same sequence of physical locations regardless of the logical address being accessed, and hence contributes to identical distribution in $\trac{\addr{1}{i}}$ and $\trac{\addr{2}{i}}$. Therefore, the distribution of cache-line-granularity accesses generated during the ORAM read phase of the $i$-th logical access is identical in executions $\exe{1}$ and $\exe{2}$.

    \bheading{Rebuilding.} We next consider the accesses generated by reshuffling. The system triggers an oblivious rebuild of the $j$-th SqrtORAM replica after $s$ logical accesses to that replica for fixed public parameters $j$ and $N$. Since the rebuild schedule depends only on the logical access index $i$ and public parameters, rebuilds are triggered at the same logical times in executions $\exe{1}$ and $\exe{2}$.
    Moreover, the reshuffling algorithm is oblivious by construction~\cite{batcher1968sorting}, meaning that the sequence of physical DRAM accesses it generates depends only on the size of the data structure and not on its contents or access history. As a result, the distribution of DRAM access traces generated by reshuffling is identical across executions $\exe{1}$ and $\exe{2}$.

    Combining the above arguments, both the ORAM read accesses and the reshuffling accesses generated by the $i$-th logical access are identically distributed in executions $\exe{1}$ and $\exe{2}$. Therefore, the overall physical DRAM access trace distribution of $\trac{\addr{1}{i}}$ is identical to that of $\trac{\addr{2}{i}}$.
\end{proof}
\end{lemma}

\begin{lemma}[Independence of Logical Access Traces]
\label{lemma:2}
The physical DRAM access traces generated by distinct logical accesses to the
vocab and merge tables are mutually independent random variables.

\begin{proof}
    Consider two distinct logical access indices $i \neq i'$, corresponding to logical addresses $\addr{}{i}$ and $\addr{}{i'}$. We show that the physical DRAM access traces generated by these accesses are independent. We distinguish two cases for the read phase and reshuffling depending on whether the accesses are served by the same SqrtORAM replica.
    
    \bheading{Read phase Case 1: Accesses served by the same SqrtORAM replica.} Suppose that  $\addr{}{i}$ and $\addr{}{i'}$ are mapped to the same SqrtORAM replica. By construction, each logical access invokes $\prp$ to determine the physical access location within the SqrtORAM. The security guarantee of SqrtORAM~\cite{oram-orig} ensures that the physical access locations corresponding to distinct logical accesses are computationally indistinguishable from independent random locations. Consequently, the physical DRAM access traces generated by accesses $\addr{}{i}$ and $\addr{}{i'}$ are independent random variables.

    \bheading{Read phase Case 2: Accesses served by different SqrtORAM replicas.} Next, suppose that $\addr{}{i}$ and $\addr{}{i'}$ are mapped to different SqrtORAM replicas. Each SqrtORAM replica is initialized independently using independent secret keys and randomness. Therefore, the $\prp$ evaluations used to select physical access locations in different SqrtORAM replicas are independent. Accordingly, the cache-line-granularity access traces generated by accesses to different SqrtORAM replicas are independent.
    
    \bheading{Rebuilding accesses.} Finally, we consider accesses generated by reshuffling. Rebuilding is triggered according to a fixed, public schedule that depends only on the number of logical accesses performed. The reshuffling algorithm is oblivious, and the sequence of cache-line-granularity accesses it generates depends only on public parameters. Consequently, reshuffling accesses are independent of the logical accesses that trigger them, as well as of reshuffling operations triggered at other times.
    
    Combining the above cases, the cache-line-granularity access traces generated by any two distinct logical accesses are independent random variables.
\end{proof}
\end{lemma}

\section{Open Science}

The paper's core contributions concern (i) obfuscating memory-access patterns of LLM tokenization (vocabulary and merge tables) using ORAM-based mechanisms, and (ii) evaluating the performance overheads relevant to serving (tokenization latency, tokenization-to-TTFT ratio, and initialization cost). 

\begin{packeditemize}
    \item \bheading{System Configuration} is introduced in the Section~\ref{sec:method}. Additionally, we use CUDA-12.4 with PyTorch 2.6.0, and flash\_attn-2.7.4+cu124torch2.6-cp312-cp312-linux\_x86\_64.whl in the Guest VM.

    \item \bheading{vLLM Modification.} 
    
    Content: Modification of vLLM to support chunked tokenization.
    
    \item \bheading{Tokenizer Implementation} 
    
    Content: ORAM-protected tokenizer library (installation is same as the original Transformer's Tokenizer library). Modify tokenizers/src/model/bpe/model.rs to support different configurations of ORAM structures.

    \item Code and artifacts are anonymously available at \url{https://anonymous.4open.science/r/oram_tokenizer-4B46}.
 
\end{packeditemize}





\end{document}